\documentclass[lettersize,journal]{IEEEtran}
\usepackage{amsmath,amsfonts}
\usepackage{amsthm}
\usepackage{adjustbox}
\usepackage{mathtools}
\usepackage{mathrsfs}
\usepackage{bm}
\usepackage{booktabs}
\usepackage{amssymb}
\usepackage{algorithm}
\usepackage{algpseudocode}
\usepackage{array}
\usepackage[caption=false,font=normalsize,labelfont=sf,textfont=sf]{subfig}
\usepackage{textcomp}
\usepackage{xcolor}
\usepackage{stfloats}
\usepackage{url}
\usepackage{verbatim}
\usepackage{graphicx}
\usepackage{cite}
\newtheorem{theorem}{Theorem}
\newtheorem{proposition}{Proposition}
\newtheorem{lemma}{Lemma}

\newtheorem{remark}{Remark}
\newtheorem{corollary}{Corollary}

\theoremstyle{definition}

\DeclareMathOperator*{\argmax}{arg\,max}
\newcommand{\C}{\mathbb{C}}

\newcommand{\tr}{\operatorname{tr}}

\newcommand{\normF}[1]{\left\lVert #1\right\rVert_{\mathrm F}}

\newcommand{\autoeq}[1]{%
  \adjustbox{max width=\columnwidth}{%
    $\displaystyle #1$
  }%
}

\begin{document}

\title{Surrogate-Enhanced Fractional Programming for MIMO Device-to-Device Interference Networks}

\author{\IEEEauthorblockN{Zihan Jiao, Xinping Yi, \textit{Member, IEEE}, Shi Jin, \textit{Fellow, IEEE}, and Giuseppe Caire, \textit{Fellow, IEEE}}

\thanks{Z. Jiao, X. Yi, and S. Jin are with the National Mobile Communications Research Laboratory, Southeast University, Nanjing 210096, China. Email:
\{jiao.zh, xyi, jinshi\}@seu.edu.cn.
}
\thanks{G. Caire is with the Faculty of Electrical Engineering and Computer Science, Technical University of Berlin, Germany. Email: 
caire@tu-berlin.de.}
}



\maketitle

\begin{abstract} Interference management in multi-stream multi-input multi-output (MIMO) device-to-device (D2D) networks often leads to weighted sum-rate maximization with sum-log-determinant involving matrix-valued signal-to-interference-plus-noise ratios (SINRs).
The state-of-the-art paradigms, including weighted minimum mean-square error (WMMSE) and fractional programming (FP), have achieved tremendous success in link scheduling, power control, and beamforming problems.
Recently, an upgraded FP approach, nicknamed surrogate-enhanced fractional programming (SEFP) in a scalar form, has attained improved sum-rate performance in joint uplink scheduling and power control in coordinated multicell single-input single-output (SISO) networks.
The proposed scalar SEFP improves the surrogate construction of the classical Lagrangian dual transform plus quadratic transform (LDT+QT) for logarithmic fractional objectives with a novel reciprocal-inverse transform (RIT), yet its extension to the matrix form for MIMO settings does not seem straightforward as in FP, because the matrix ratio and the auxiliary matrix generally do not commute.
In this paper, by leveraging the mathematical tools from Hermitian functional calculus, we extend RIT from scalar to matrix ratios by lifting the scalar RIT along the eigendirections of a Hermitian auxiliary matrix and recasting the resulting directional construction in an operator form.
Combining the matrix RIT with the matrix QT, we develop a matrix SEFP framework for weighted sum-log-determinant maximization.
Further, we establish a unified view of matrix SEFP and the recently proposed XMMSE method, specify when the two algorithms attribute to identical minorization-maximization (MM) surrogates and original variable update trajectories, and prove XMMSE can be considered as a special case of SEFP under the algorithmic family perspective. Inspired by the unified view, we develop SEFPLinQ for the joint scheduling and multi-stream beamforming optimization in flexible-association MIMO D2D networks.
Numerical results demonstrate consistent performance gains of the proposed SEFPLinQ over FPLinQ and other baselines in both SISO and MIMO settings.

\end{abstract}

\begin{IEEEkeywords}
Fractional programming (FP), surrogate-enhanced FP (SEFP), weighted sum-rate (WSR), minorization-maximization (MM), beamforming, scheduling.
\end{IEEEkeywords}

\section{Introduction}
Interference management is a fundamental problem for the design of modern multi-terminal wireless networks such as the multi-stream multi-input multi-output (MIMO) device-to-device (D2D) interference networks with the objective of maximizing weighted sum-rate (WSR)~\cite{WSRReview}. 
With interference management techniques such as link scheduling, power control, and MIMO beamforming, the WSR maximization problem is notoriously difficult because the desired signal of one link is generally coupled with the interference experienced by many others, leading to highly nonconvex optimization landscapes, and is generally known to be NP-hard~\cite{MISOComplexity}. 
The difficulty becomes more pronounced in multi-antenna multi-stream transmission networks, where the achievable rates naturally take a log-determinant formula involving matrix-form signal-to-interference-plus-noise ratios (SINRs). 

To address these challenges, existing approaches to interference-aware scheduling and resource allocation can be broadly organized into several methodological lines.
One line of work emphasizes low-complexity link scheduling based on certain interference compatibility criteria. 
For example, FlashLinQ~\cite{FlashLinQ} employed distributed signal-to-interference ratio (SIR)-based coordination, whereas ITLinQ~\cite{ITLinQ} and ITLinQ+~\cite{ITLinQplus} incorporated information theory-motivated treating interference as noise (TIN) conditions to guide link scheduling and, in the latter case, power control.
A complementary line directly optimizes the WSR function. For instance, the WMMSE framework exploited the rate-MSE relationship to recast WSR beamforming into tractable alternating block updates, while related block coordinate methods have further incorporated scheduling, beamforming, and power adaptation~\cite{WMMSE_1,WMMSE_2,BCD}.
Notably, fractional programming (FP) provides another optimization-based route by decoupling the signal and interference terms in the SINR-type ratios: the quadratic transform (QT) based framework first addressed power control and beamforming~\cite{FPPart1}, was subsequently extended to discrete scheduling through matching~\cite{FPPart2}, and was further generalized by matrix FP/FPLinQ to joint scheduling and multi-stream beamforming in MIMO D2D networks with flexible association~\cite{FPPart3}. 
More recently, artificial intelligent (AI) induced methods, such as end-to-end graph-based learning~\cite{2023TWC_PCGNN}, deep unfolding~\cite{jiao2026ISIT}, and hybrid model/data-driven~\cite{GRLinQ} methods have been explored to reduce the online computational and signaling burdens of such iterative optimization procedures.

Of particular relevance to this work is the model-based WSR optimization, where the block coordinate ascent (BCA) and the minorization-maximization (MM) viewpoints provide a common language for understanding these methods.
Once the transform-induced auxiliary variables are fixed at their exact maximizers, the transformed objective acts as a globally valid minorizing surrogate that is tight at the current iterate~\cite{FPPart3,BCA+MM WSR}.
This perspective explains why apparently different WMMSE-, FP-, and MM-type iterations can coincide under particular auxiliary variable coordinates and update orders, while also revealing that neither the ratio decoupling nor the block update strategy is unique.
More importantly, it places surrogate construction at the center of algorithm design: An effective surrogate should remain tractable enough to preserve exploitable problem structure, while being sufficiently faithful to the original objective to guide the optimization variables efficiently~\cite{MM tutorial}.

Motivated by this perspective, the recently proposed surrogate-enhanced fractional programming (SEFP) framework revisited the classical LDT+QT construction in FP and introduced the reciprocal-inverse transform (RIT), which yields a tighter surrogate for logarithmic fractional objectives while retaining the favorable decoupling structure of FP~\cite{SEFP_scalar}.
Nevertheless, the existing SEFP framework was formulated for scalar ratios and its extension to matrix-valued SINR ratios is nontrivial.
Along a different but closely related direction, the most recent XMMSE method exploited an exact integral rate-MSE relation to construct a tighter MM surrogate than conventional WMMSE for MIMO beamforming~\cite{XMMSE}.
This central question then arises as to whether the surrogate enhancement principle underlying SEFP admits a principled extension from scalar fractional objectives to weighted sum-log-determinant optimization with matrix-valued ratios, and what the relationship is between the SEFP-type and the XMMSE-type surrogate enhancement.



To answer the above question, we extend the scalar SEFP to its matrix counterpart to deal with matrix-valued ratios, and establish the connection to the recently proposed XMMSE.
Specifically, the main contributions of this paper are summarized as follows.
\begin{itemize}
    \item 
    By leveraging the mathematical tools from Hermitian functional calculus, we develop the matrix RIT by lifting the scalar RIT along the eigendirections of a Hermitian auxiliary matrix and recasting the resulting directional construction in an operator form. By combining the matrix RIT with the matrix QT, we develop a matrix SEFP framework for weighted sum-log-determinant maximization with matrix-valued ratios.
    \item 
    We establish a unified view of matrix SEFP and the recently proposed XMMSE method. On the one hand, we show that the auxiliary variables induced by the transform of the two algorithms admit a lossless finite-continuum-dimensional interconvertibility, and, under exact auxiliary-variable updates, the two approaches induce identical MM surrogates. On the other hand, from an algorithmic perspective, however, SEFP retains additional flexibility in ratio decoupling and block-coordinate update ordering beyond those available in XMMSE.  
    \item 
    By exploiting the above flexibility, we develop SEFPLinQ for the joint scheduling and multi-stream beamforming optimization in flexible-association MIMO D2D networks, where the SEFP surrogate admits an edge-separable structure that enables a two-stage weighted bipartite matching procedure.
    Numerical results demonstrate consistent performance gains of the proposed SEFPLinQ over FPLinQ and other baselines in both SISO and MIMO settings.
\end{itemize}
The rest of this paper is organized as follows.
Section~\ref{preliminary} reviews the scalar RIT and SEFP and lays the mathematical groundwork for their matrix extension.
Section~\ref{sec:matrix_rit} develops the matrix RIT and SEFP. 
Section~\ref{sec:xmmse_connection} establishes an interconnection between the SEFP and XMMSE, and clarifies their equivalence and difference.
Section~\ref{SEFPLinQ} proposes the SEFPLinQ strategy for the joint scheduling and beamforming WSR optimization problem.
Section~\ref{sec:numerical results} presents the numerical results, and Section~\ref{sec:conclu} concludes the paper.

\textbf{Notation:} Boldface, lowercase, and uppercase letters denote vectors and matrices, respectively. 
$(\cdot)^{\dagger}$, $(\cdot)^{T}$, and $(\cdot)^{H}$ represent the pseudo-inverse, transpose and Hermitian (conjugate) transpose operators, respectively.
$\mathbb{C}^{n}$, $\mathbb{H}_{+}^{n}$ and $\mathbb{H}_{++}^{n}$ denote the $n \times n$ complex, Hermitian positive semidefinite and positive definite matrices, respectively.
$\mathbb{R}_{+}$ and $\mathbb{R}_{++}$ denotes nonnegative and positive real number.
For Hermitian matrices $\mathbf{X}$ and $\mathbf{Y}$,
$\mathbf{X}\succeq\mathbf{Y} \ (\mathbf{X}\succ\mathbf{Y})$ means that $\mathbf{X}-\mathbf{Y}$ is positive semidefinite (positive definite).
$\operatorname{tr}(\cdot)$ denotes the trace of a square matrix.
$\langle\mathbf{X}, \mathbf{Y}\rangle_F \triangleq \operatorname{tr}(\mathbf{X}^{H}\mathbf{Y})$ denotes the Frobenius inner product between matrices $\mathbf{X}$ and $\mathbf{Y}$, and $\Vert{}\mathbf{X}\Vert{}_F = \sqrt{\langle\mathbf{X}, \mathbf{X}\rangle_F}$ is the corresponding Frobenius norm. 
$\mathbf{I}_N$ is the $N \times N$ identity matrix.
$\ker(\mathbf{A})$ denotes the kernel (or null space) of matrix $\mathbf{A}$.
We use the underline $\underline{(\cdot)}$ to denote a collection of vector and matrix variables.

\section{Preliminaries}\label{preliminary}
The weighted sum-rate maximization is usually formulated, in a scalar form, as
a generic weighted sum-of-logarithms maximization problem 
\begin{equation}\label{w_sol_p}
\max_{\mathbf{x} \in \mathcal{X}} \quad F_s(\mathbf{x}) \triangleq \sum_{m=1}^{M} \omega_m \log \left( 1 + \frac{A_m(\mathbf{x})}{B_m(\mathbf{x})} \right),\end{equation}
with $\omega_m\geq 0$, where $\mathbf{x}$ collects the optimization variables and $\mathcal X$ denotes the feasible set.  $A_m(\mathbf{x})\in\mathbb{R}_{+}$ and $B_m(\mathbf{x})\in\mathbb{R}_{++}$ are the numerator and the denominator functions, respectively. 

The state-of-the-art solution to \eqref{w_sol_p} is the fractional programming (FP) method \cite{FPPart1,FPPart2}, which consists of a Lagrangian dual transform (LDT) dealing with sum-of-logarithms followed by a quadratic transform (QT) dedicated to the ratio. The FP algorithm has been effective for power control, beamforming, and uplink scheduling in single-antenna D2D networks with fixed or flexible transmitter-receiver associations.

When it comes to the case with multiple antennas at both transmitters and receivers, the generic weighted sum-log-determinant maximization problem in matrix form is
\begin{equation}\label{prob:generic_logdet}
\max_{\mathbf{x} \in \mathcal{X}}\quad
F(\mathbf{x})\triangleq
\sum_{m=1}^{M}\omega_m
\log\det\!\left(\mathbf{I}_{d_m}+\mathbf{R}_m(\mathbf{x})\right),
\end{equation}
where $\mathbf{x}$ collects the optimization variables (e.g., beamforming matrices) that belong to the feasible set $\mathcal X$. 
Notably, the ratio becomes a matrix form, i.e., $
\mathbf{R}_m(\mathbf{x})
\triangleq
\mathbf{S}_m^{H}(\mathbf{x})\mathbf{F}_m^{-1}(\mathbf{x})\mathbf{S}_m(\mathbf{x}) \
\in\mathbb{H}_{+}^{d_m},$
with $\mathbf{S}_m(\mathbf{x})\in\mathbb{C}^{n_m\times d_m}$ and
$\mathbf{F}_m(\mathbf{x})\in\mathbb{H}_{++}^{n_m}$, where $d_m$ is the number of data streams in MIMO communications. For compactness, we omit $\mathbf{x}$ and the index $m$ whenever they are clear from the context.
With the FP framework, the solution to~\eqref{prob:generic_logdet}  with multi-antenna and multi-data-stream transmission is a natural extension of the scalar form in~\eqref{w_sol_p} with $d_m=1$, and has been applied to the MIMO D2D networks in a straightforward way \cite{FPPart3}.

Notably, the scalar FP framework~\cite{FPPart1,FPPart2} has recently been enhanced in \cite{SEFP_scalar}, nicknamed SEFP, where the LDT in the former is replaced by a newly designed RIT in the latter.
The RIT produces a pointwise tighter surrogate function than the LDT from a MM perspective, in such a way that the SEFP yields superior performance for joint uplink scheduling and power control in the SISO D2D networks.


\subsection{The Scalar RIT and SEFP \cite{SEFP_scalar}}
\label{subsec:scalar_rit_review}
In what follows, we briefly introduce the newly designed RIT and the SEFP algorithm in the scalar form.

\begin{proposition}[Scalar Reciprocal-Inversion Transform~{\cite[Theorem~1]{SEFP_scalar}}]\label{Scalar_RIT}
The weighted sum-of-logarithms problem \eqref{w_sol_p} can be
equivalently reformulated as
\begin{equation}
\label{eq:scalar_rit_problem}
\max_{\substack{\mathbf{x}\in\mathcal{X}\\
\boldsymbol{\alpha}\in\mathbb{R}_{++}^{M}}}
\quad
F_{ r}(\mathbf{x},\boldsymbol{\alpha}) \triangleq
\sum_{m=1}^{M}
\omega_m
\frac{
c(\alpha_m)A_m(\mathbf{x})
}{
\alpha_m^2 B_m(\mathbf{x})
+
b(\alpha_m)A_m(\mathbf{x})
},
\end{equation}
where
$\boldsymbol{\alpha}=(\alpha_1,\ldots,\alpha_M)$ denotes the collection
of RIT-induced auxiliary variables, and two scalar functions for $v>0$ are defined as
\begin{equation} \label{operator_b} 
\quad b(v)\triangleq (1 + v)\log(1+v) - v,
\end{equation} 
\begin{equation} \label{operator_c}
c(v)\triangleq \left( 1 + v \right) \log^2(1+v).
\end{equation}

For any fixed $\mathbf{x}$, if $A_m(\mathbf{x})>0$, the unique optimal
auxiliary variable associated with the $m$-th term is $
\alpha_m^\star=
{A_m(\mathbf{x})}/{B_m(\mathbf{x})}.$
If $A_m(\mathbf{x})=0$, the corresponding $m$-th term of $F_r(\mathbf{x},\boldsymbol{\alpha})$ equals zero for any $\alpha_m > 0$, and hence exactly matches the $m$-th term of $F_s(\mathbf{x})$. 
\end{proposition}

By the scalar RIT in Proposition~\ref{Scalar_RIT}, we have
\begin{equation}
\label{eq:scalar_rit_pointwise_equiv}
\max_{\boldsymbol{\alpha}\in\mathbb{R}_{++}^{M}}
F_{r}(\mathbf{x},\boldsymbol{\alpha})
=
F_s(\mathbf{x}), \quad
\forall\mathbf{x}\in\mathcal{X}.
\end{equation}
%
According to \cite{SEFP_scalar}, the RIT serves as a pointwise lower-bound surrogate of the logarithm $\log(1+r)$ when setting $r=\frac{A_m(\mathbf{x})}{B_m(\mathbf{x})}$.


\begin{corollary}[Scalar RIT's Lower-Bound Property~{\cite[Corollary~1]{SEFP_scalar}}]
\label{cor:scalar_rit_lower_bound}
For every fixed $\alpha>0$, we have
\begin{equation}
\label{eq:scalar_rit_lower_bound}
\log(1+r) \ge \phi_{\alpha}(r)\triangleq
\frac{c(\alpha)r}
{\alpha^2+b(\alpha)r},
\quad
r\geq 0.
\end{equation}
For $r>0$, equality in \eqref{eq:scalar_rit_lower_bound} holds if and
only if $\alpha=r$, whereas, at $r=0$, the equality always holds.
\end{corollary}

We next recall the scalar SEFP formulation obtained by concatenating the scalar RIT with the scalar QT.
\begin{proposition}[Scalar Surrogate-Enhanced FP{~\cite[Theorem~2]{SEFP_scalar}}] \label{Scalar_SEFP}
By applying the scalar QT \cite[Theorem 1]{FPPart1} to the
RIT-reformulated objective in \eqref{eq:scalar_rit_problem}, problem
\eqref{w_sol_p} can be equivalently reformulated as
\begin{equation}
\begin{aligned}
\label{eq:scalar_sefp_problem}
\max_{\substack{
\mathbf{x}\in\mathcal{X}\\
\boldsymbol{\alpha}\in\mathbb{R}_{++}^{M}\\
\mathbf{y}\in\mathbb{R}^{M}
}}
\
F_{rq}(\mathbf{x},\boldsymbol{\alpha},\mathbf{y})&\triangleq
\sum_{m=1}^{M}
\Big[
2y_m
\sqrt{
\omega_m c(\alpha_m)A_m(\mathbf{x})
}
\\[-12pt]
&-
y_m^2
\big(
\alpha_m^2B_m(\mathbf{x})
+
b(\alpha_m)A_m(\mathbf{x})
\big)
\Big],
\end{aligned}
\end{equation}
where $\mathbf{y}=(y_1,\ldots,y_M)$ denotes the collection of
QT-induced auxiliary variables.
For any fixed $(\mathbf{x},\boldsymbol{\alpha})$, the optimal QT
auxiliary variables are given by
\begin{equation}
\label{eq:scalar_y_opt}
y_m^\star
=
\frac{
\sqrt{\omega_m c(\alpha_m)A_m(\mathbf{x})}
}{
\alpha_m^2B_m(\mathbf{x})
+
b(\alpha_m)A_m(\mathbf{x})
},
\quad
m=1,\ldots,M.
\end{equation}
\end{proposition}

By Proposition \ref{Scalar_SEFP}, we have
\begin{equation}
\label{eq:scalar_sefp_pointwise_equiv}
\max_{\substack{
\boldsymbol{\alpha}\in\mathbb{R}_{++}^{M}\
\mathbf{y}\in\mathbb{R}^{M}
}}
F_{rq}
(\mathbf{x},\boldsymbol{\alpha},\mathbf{y})
=
F_s(\mathbf{x}),
\quad
\forall\mathbf{x}\in\mathcal{X}.
\end{equation}
This means the solution to the original weighted sum-of-logarithms maximization problem in \eqref{w_sol_p} can be alternatively done by solving the equivalent formulation in \eqref{eq:scalar_sefp_problem} with two additional auxiliary variables $\boldsymbol{\alpha}$ and $\mathbf{y}$.

When extending the scalar SEFP to the matrix case, we would expect a similar straightforward step as that for the FP \cite{FPPart3}. However, it turns out that the scalar-to-matrix extension for the SEFP is totally nontrivial.

\subsection{Why is the Scalar-to-Matrix Extension Nontrivial?}\label{rem:noncommutativity}
The main obstruction to the extension is noncommutativity. 
In contrast to the scalar-to-matrix extension of FP, the matrix RIT can not be obtained by mechanically replacing the scalar auxiliary variable $\boldsymbol{\alpha}$ with a matrix auxiliary variable $\boldsymbol{\Gamma}$.\footnote{The scalar-to-matrix extension of the LDT, however, exactly uses this method (see in~\cite[Prop.~3 and Thm.~2]{FPPart3}).} 
The hypothetical transformed objective expression is not a valid representation because the matrix products involved in the resulting rational expression cannot, in general, be rearranged as in the scalar case.
This is because $\boldsymbol{\Gamma}$ and $\mathbf{R}$ need not be simultaneously diagonalizable to satisfy the commutativity.

To invite deeper consideration, note that $\phi_{\alpha}(r)$ in \eqref{eq:scalar_rit_lower_bound} is a lower bound of $\log(1+r)$ that is tight at $r=\alpha$ for $r>0$.
Hence, the scalar auxiliary variable only needs to specify a single tangency value. 
For a matrix ratio, however, let $\mathbf R\in\mathbb H_+^d$ have eigenvalues $\lambda_i(\mathbf R)$, $i=1,\ldots,d$, since 
$\log\det(\mathbf{I}_d+\mathbf{R})
=\sum_{i=1}^{d}\log(1+\lambda_i(\mathbf{R})), $
the matrix extension of the RIT must associate potentially different scalar contact values with different $d$-dimensional spectral directions. 
Accordingly, the potential auxiliary matrix $\boldsymbol{\Gamma}$ needs to be chosen Hermitian positive definite, with eigendecomposition
\begin{equation}\label{eq:gamma_evd}
\boldsymbol{\Gamma}
=\boldsymbol{\Xi}  \operatorname{diag}(\gamma_1,\ldots,\gamma_d) \boldsymbol{\Xi}^{H},\ \quad \boldsymbol{\Xi}=[\boldsymbol{\xi}_1,\ldots,\boldsymbol{\xi}_d],
\end{equation}
where the eigenvalues $\{\gamma_i\}_{i=1}^{d}$ play the role of scalar contact values, whereas the eigenvectors $\{\boldsymbol{\xi}_i\}_{i=1}^{d}$ specify the spectral directions along which the scalar lower-bound property is lifted. This becomes feasible when resorting to the Hermitian functional calculus. 

\subsection{Hermitian Functional Calculus}
\label{subsec:Hermitian functional calculus}
We introduce a widely utilized mathematical tool--Hermitian functional calculus--in this subsection, which is critical for the following matrix RIT and SEFP constructions and the SEFP-XMMSE interconnection analysis.

For a Hermitian matrix $\mathbf{X}=\mathbf{U}\operatorname{diag}(\lambda_1,\ldots,\lambda_N)\mathbf{U}^{H}$,
let $\sigma(\mathbf{X})=\{\lambda_1,\ldots,\lambda_N\}$
denote its set of eigenvalues, called its spectrum.
For a scalar function $f$ which is well defined on $\sigma(\mathbf{X})$, the corresponding matrix function is defined as~\cite[Sec.~10.3]{shapiro2024functional}, \cite[Sec.~12.4.1]{einsiedler2017functional}
$$ f(\mathbf{X})\triangleq
\mathbf{U}\operatorname{diag}\!\big(f(\lambda_1),\ldots,f(\lambda_N)\big)\mathbf{U}^{H}. $$
This definition is independent of the eigenbasis chosen within each repeated eigenvalue subspace.

For scalar functions $f,g$, which are both well defined on $\sigma(\mathbf{X})$ and scalars $\alpha,\beta\in\mathbb{C}$, functional calculus preserves the algebraic operations for the corresponding matrix functions~\cite[Thm.~12.37]{einsiedler2017functional}
\[
\begin{aligned}
(\alpha f+\beta g)(\mathbf{X})
&=\alpha f(\mathbf{X})+\beta g(\mathbf{X}),\\
(fg)(\mathbf{X})
&=f(\mathbf{X})g(\mathbf{X})=g(\mathbf{X})f(\mathbf{X}).
\end{aligned}
\]
Thus, functions of the same Hermitian matrix share a common orthonormal eigenbasis and commute. 
Functions of different matrices need not commute.

Using the same definitions of the scalar functions 
$f,g$ and of the Hermitian matrix $\mathbf{X}$ as in the preceding two paragraphs, then
$f(\mathbf{X})$ and $g(\mathbf{X})$ are Hermitian, and pointwise scalar inequalities yield corresponding matrix inequalities~\cite[Thm.~12.37 and Cor.~12.42]{einsiedler2017functional}
\[
f(t)\leq g(t),\ \forall t\in\sigma(\mathbf{X}) \quad\Longrightarrow\quad f(\mathbf{X})\preceq g(\mathbf{X}).
\]
In particular, nonnegative or strictly positive values on the spectrum yield positive semidefinite or positive definite matrix functions, respectively.
This property could also yield a matrix RIT  construction as a lower bound of the logarithm function as in Lemma~\ref{lem:mRIT_LB} in the Appendix.

For a fixed $\mathbf{X}$, parameter integration can also be performed eigenvalue-wise~\cite[Example~12.43]{einsiedler2017functional}.
Specifically, if $h_\tau(t)$ is integrable in $\tau\in[0,1]$ at each $t\in\sigma(\mathbf{X})$, then $h(t)\triangleq\int_0^1 h_\tau(t)\,d\tau$ satisfies $h(\mathbf{X})=\int_0^1 h_\tau(\mathbf{X})\,d\tau$.
This property enables the spectral moment identities used to connect SEFP and XMMSE in Section~\ref{sec:xmmse_connection}.


\section{Matrix Reciprocal-Inversion Transform and Surrogate-Enhanced Fractional Programming}
\label{sec:matrix_rit}
With Hermitian functional calculus, we are able to extend the scalar RIT and the SEFP framework to the matrix version. 

\subsection{The Matrix RIT}
\label{subsec:matrix_rit}
With the matrix functional calculus introduced in Section~\ref{subsec:Hermitian functional calculus}, 
the scalar functions $b(\cdot)$ and $c(\cdot)$ in
\eqref{operator_b} and \eqref{operator_c} induce the corresponding Hermitian matrix functions.
For $\boldsymbol{\Gamma}$ in \eqref{eq:gamma_evd}, define
\begin{align}
b(\boldsymbol{\Gamma})
&\triangleq
(\mathbf{I}+\boldsymbol{\Gamma})\log(\mathbf{I}+\boldsymbol{\Gamma})-\boldsymbol{\Gamma},\label{eq:b_matrix}\\
c(\boldsymbol{\Gamma})
&\triangleq
(\mathbf{I}+\boldsymbol{\Gamma})[\log(\mathbf{I}+\boldsymbol{\Gamma})]^2.\label{eq:c_matrix}
\end{align}
By Hermitian functional calculus, we have
\begin{align}
b(\boldsymbol{\Gamma})
&=\boldsymbol{\Xi}\operatorname{diag}\big(b(\gamma_1),\ldots,b(\gamma_d)\big)\boldsymbol{\Xi}^{H},\label{eq:b_spectral}\\
c(\boldsymbol{\Gamma})
&=\boldsymbol{\Xi}\operatorname{diag}\big(c(\gamma_1),\ldots,c(\gamma_d)\big)\boldsymbol{\Xi}^{H}.\label{eq:c_spectral}
\end{align}
Since $b(\gamma_i)>0$ and $c(\gamma_i)>0$ for every $\gamma_i>0$, both
$b(\boldsymbol{\Gamma})$ and $c(\boldsymbol{\Gamma})$ are positive definite.

For each $m$, introduce $\boldsymbol{\Gamma}_m\in\mathbb{H}_{++}^{d_m}$ and define
\begin{subequations}\label{eq:vec_definitions}
\begin{align}
\mathbf{a}_{m,\boldsymbol{\Gamma}_m}
&\triangleq
\operatorname{vec}\!\left(\mathbf{S}_m c(\boldsymbol{\Gamma}_m)^{1/2}\right),\label{eq:a_gamma}\\
\mathbf{D}_{m,\boldsymbol{\Gamma}_m}
&\triangleq
(\boldsymbol{\Gamma}_m^{2})^{T}\otimes\mathbf{F}_m
+b(\boldsymbol{\Gamma}_m)^{T}\otimes\mathbf{S}_m\mathbf{S}_m^{H}.
\label{eq:D_gamma_kron}
\end{align}
\end{subequations}
We are now ready to state the matrix RIT.
\begin{theorem}[Matrix Reciprocal-Inversion Transform]\label{thm:matrix_rit}
The weighted sum-log-determinant problem \eqref{prob:generic_logdet} is equivalent to
\begin{equation}\label{eq:matrix_rit_problem}
\max_{\substack{\mathbf{x}\in\mathcal{X} , \{\boldsymbol{\Gamma}_m\in\mathbb{H}_{++}^{d_m}\}_{m=1}^{M}}}
F_{R}(\mathbf{x},\underline{\boldsymbol{\Gamma}})\triangleq
\sum_{m=1}^{M}\omega_m\,
\mathbf{a}_{m,\boldsymbol{\Gamma}_m}^{H}
\mathbf{D}_{m,\boldsymbol{\Gamma}_m}^{-1}
\mathbf{a}_{m,\boldsymbol{\Gamma}_m},
\end{equation}
where for every positively weighted term with \(\mathbf{R}_m\succ\mathbf{0}\), the unique optimal RIT-induced auxiliary matrix is
\begin{equation}\label{eq:gamma_opt_pd}
\boldsymbol{\Gamma}_m^{\star}
=
\mathbf{R}_m
=
\mathbf{S}_m^{H}\mathbf{F}_m^{-1}\mathbf{S}_m.
\end{equation}
If otherwise $\mathbf{R}_m$ is singular with eigen-decomposition
\begin{equation} 
\label{eq:Rm_decom}
\mathbf{R}_m =
\mathbf{P}_m
\operatorname{diag}\big(
\lambda_{m,1},\ldots,\lambda_{m,\kappa_m},
0,\ldots,0
\big)
\mathbf{P}_m^{H}, \ \lambda_{m,i}>0, 
\end{equation}
then every matrix of the form
\begin{equation} \label{eq:gamma_opt_singular}
\boldsymbol{\Gamma}_m^{\star} =
\mathbf{P}_m \operatorname{diag} \big(
\lambda_{m,1},\ldots,\lambda_{m,\kappa_m},
\eta_{m,\kappa_m+1},\ldots,\eta_{m,d_m}
\big)
\mathbf{P}_m^{H},
\end{equation}
for any $\eta_{m,j}>0$ is optimal and can attain the corresponding maximum term in~\eqref{eq:matrix_rit_problem}.
\end{theorem}
\begin{proof}
See Appendix~\ref{pro:mRIT_theorem}.
\end{proof}

\begin{remark}[Scalar RIT as a special case]\label{rem:scalar_reduction} 
When $d_m=1$, let $\boldsymbol{\Gamma}_m=\gamma_m$, $\mathbf{S}_m=\mathbf{s}_m$, and $r_m=\mathbf{s}_m^{H}\mathbf{F}_m^{-1}\mathbf{s}_m$. Then
\begin{equation}
\mathbf{a}_{m,\gamma_m}^{H}
\mathbf{D}_{m,\gamma_m}^{-1}
\mathbf{a}_{m,\gamma_m}
=
\frac{c(\gamma_m)r_m}
{\gamma_m^2+b(\gamma_m)r_m},
\end{equation}
so Theorem~\ref{thm:matrix_rit} reduces exactly to the scalar RIT in {Proposition~\ref{Scalar_RIT}}. 
\end{remark}

To facilitate the successive combination with the QT in Section~\ref{subsec:matrix_sefp}, we give another equivalent representation of the matrix RIT by the following corollary.
\begin{corollary} [Operator-form Representation of Matrix RIT] \label{cor:mRIT_equ_repre}
For every \(\mathbf{F}_m\succ\mathbf{0}\) and \(\boldsymbol{\Gamma}_m\succ\mathbf{0}\),
define the linear matrix operator
\begin{equation}\label{eq:Dm_operator}
\mathcal{D}_{m,\boldsymbol{\Gamma}_m}(\mathbf{Y})
\triangleq
\mathbf{F}_m\mathbf{Y}\boldsymbol{\Gamma}_m^2
+\mathbf{S}_m\mathbf{S}_m^{H}\mathbf{Y}b(\boldsymbol{\Gamma}_m).
\end{equation}
The matrix RIT-induced objective $F_{R}$ defined in~\eqref{eq:matrix_rit_problem} in Theorem~\ref{thm:matrix_rit} can be equivalently represented by 
\begin{equation}
\autoeq{
F_{R}(\mathbf{x},\underline{\boldsymbol{\Gamma}})
\triangleq\sum_{m=1}^{M}\omega_m
\left\langle
\mathbf{S}_m c(\boldsymbol{\Gamma}_m)^{1/2},
\mathcal{D}_{m,\boldsymbol{\Gamma}_m}^{-1}
\big(\mathbf{S}_m c(\boldsymbol{\Gamma}_m)^{1/2}\big)
\right\rangle_F.}
\label{eq:operator_kron_equivalence}
\end{equation}
\end{corollary}
\begin{proof}
See Appendix~\ref{pro:coro_oper_repr}.
\end{proof}

\subsection{Matrix SEFP = Matrix RIT + QT}
\label{subsec:matrix_sefp}
As in the scalar case~\cite[Section III.C]{SEFP_scalar}, the matrix RIT remains compatible with the matrix   QT~\cite[Theorem~1]{FPPart3}. 
To apply the QT directly to the operator-form representation of the matrix RIT in Corollary~\ref{cor:mRIT_equ_repre}, we first state the QT identity in an equivalent operator form in the following lemma.

\begin{lemma}[Operator-Form QT Identity]\label{lem:operator_qt}
Let $\mathcal{D}:\mathbb{C}^{n\times d}\rightarrow\mathbb{C}^{n\times d}$ be a self-adjoint and positive definite linear operator. 
Then, for any $\mathbf{A}\in\mathbb{C}^{n\times d}$,
\begin{equation}\label{eq:operator_qt_identity}
\left\langle
\mathbf{A},\mathcal{D}^{-1}(\mathbf{A})
\right\rangle_F
=
\max_{\mathbf{Y}\in\mathbb{C}^{n\times d}}
\left\{
2\Re\left\langle\mathbf{A},\mathbf{Y}\right\rangle_F
-
\left\langle\mathbf{Y},\mathcal{D}(\mathbf{Y})\right\rangle_F
\right\}.
\end{equation}
The optimal QT-induced auxiliary variable is unique and is given by
\begin{equation}\label{eq:operator_qt_opt}
\mathbf{Y}^{\star}=\mathcal{D}^{-1}(\mathbf{A}).
\end{equation}
\end{lemma}
\begin{proof}
Applying $\mathcal{D}$ to both sides of~\eqref{eq:operator_qt_opt} gives $\mathcal{D}(\mathbf{Y}^{\star})=\mathbf{A}$. 
Since $\mathcal{D}$ is self-adjoint, for any $\mathbf{Y}$,
\begin{align}
&2\Re\left\langle\mathbf{A},\mathbf{Y}\right\rangle_F
-\left\langle\mathbf{Y},\mathcal{D}(\mathbf{Y})\right\rangle_F
\nonumber\\
&\quad=
\left\langle
\mathbf{A},\mathcal{D}^{-1}(\mathbf{A})
\right\rangle_F
-\left\langle
\mathbf{Y}-\mathbf{Y}^{\star},
\mathcal{D}(\mathbf{Y}-\mathbf{Y}^{\star})
\right\rangle_F.
\label{eq:operator_qt_complete_square}
\end{align}
Because $\mathcal{D}$ is positive definite, the second term on the right-hand side of \eqref{eq:operator_qt_complete_square} is nonnegative and equals zero if and only if $\mathbf{Y}=\mathbf{Y}^{\star}$. 
Therefore, \eqref{eq:operator_qt_identity} holds, with the unique optimizer given by \eqref{eq:operator_qt_opt}.
\end{proof}
Applying Lemma~\ref{lem:operator_qt} to the matrix RIT leads to the following matrix SEFP formulation.
\begin{theorem}[Matrix Surrogate-Enhanced Fractional Programming]\label{thm:matrix_sefp}
By applying the matrix RIT and the matrix QT successively, the weighted sum-log-determinant problem \eqref{prob:generic_logdet} is equivalent to
\begin{subequations}\label{prob:matrix_sefp}
\begin{align}
&\underset{\mathbf{x},\{\boldsymbol{\Gamma}_m,\mathbf{Y}_m\}_{m=1}^{M}}
{\operatorname{max}}\quad F_{RQ}(\mathbf{x},\underline{\boldsymbol{\Gamma}},\underline{\mathbf{Y}})
\label{prob:matrix_sefp_obj}\\
&\qquad \mathrm{s.t.}\quad
\mathbf{x}\in\mathcal{X},\ 
\boldsymbol{\Gamma}_m\in\mathbb{H}_{++}^{d_m},\ 
\mathbf{Y}_m\in\mathbb{C}^{n_m\times d_m},
\label{prob:matrix_sefp_aux}
\end{align}
\end{subequations} 
\begin{figure*}[!t]
\normalsize
\begin{equation} \label{def_of_FmSEFP}
 F_{RQ}(\mathbf{x},\underline{\boldsymbol{\Gamma}},\underline{\mathbf{Y}})   \triangleq
\sum_{m=1}^{M}
\Big[
2\sqrt{\omega_m}\,\Re\operatorname{tr}\!\left(c(\boldsymbol{\Gamma}_m)^{1/2}\mathbf{S}_m^{H}\mathbf{Y}_m
\right)
-\operatorname{tr}\!\left(
\mathbf{Y}_m^{H}\mathbf{F}_m\mathbf{Y}_m\boldsymbol{\Gamma}_m^{2}
\right)
-\operatorname{tr}\!\left(
\mathbf{Y}_m^{H}\mathbf{S}_m\mathbf{S}_m^{H}\mathbf{Y}_m b(\boldsymbol{\Gamma}_m)
\right)
\Big].
\end{equation}
\hrulefill 
\vspace*{-1pt} 
\end{figure*}
where $F_{RQ}(\mathbf{x},\underline{\boldsymbol{\Gamma}},\underline{\mathbf{Y}})$ is defined in \eqref{def_of_FmSEFP}.

For any fixed $\mathbf{x}$, the optimal RIT-induced auxiliary matrix $\boldsymbol{\Gamma}_m^{\star}$ is characterized exactly as in Theorem~\ref{thm:matrix_rit}.
Then, for any fixed $\mathbf{x}$ and $\boldsymbol{\Gamma}_m$, Lemma~\ref{lem:operator_qt} gives the unique optimal QT-induced auxiliary variable update 
\begin{equation}\label{eq:Y_operator_update}
\mathbf{Y}_m^{\star} =
\sqrt{\omega_m}\,
\mathcal{D}_{m,\boldsymbol{\Gamma}_m}^{-1}
\!\left(
\mathbf{S}_m c(\boldsymbol{\Gamma}_m)^{1/2}
\right).
\end{equation}
Furthermore, when $\boldsymbol{\Gamma}_m$ is updated optimally, $\mathbf{Y}_m^{\star}$ has the closed-form expression
\begin{equation}\label{eq:Y_closeform}
\mathbf{Y}_m^{\star} =\sqrt{\omega_m} \mathbf{F}^{-1}_m \mathbf{S}_m \mathbf{R}_m^\dagger(\mathbf{I}+\mathbf{R}_m)^{-1/2}.
\end{equation}
\end{theorem}

\begin{proof} See Appendix~\ref{proof_of_Thm_mSEFP}.
\end{proof}

\subsection{MM Interpretation}
\label{subsec:mm_interpretation}
According to the correspondence between the BCA algorithm and MM theory~\cite[Section V.A]{BCA+MM WSR}, it follows that a BCA algorithm's objective with the auxiliary variables fixed at their current optimal values is a globally valid minorizing surrogate of the original objective and is tight at the current iterate. As such, the matrix SEFP also admits a direct MM interpretation. 

Specifically, by Theorems~\ref{thm:matrix_rit} and~\ref{thm:matrix_sefp},
\begin{equation}\label{eq:mm_value_representation}
F(\mathbf{x})
=
\max_{\underline{\boldsymbol{\Gamma}}}
F_R(\mathbf{x},\underline{\boldsymbol{\Gamma}})
=
\max_{\underline{\boldsymbol{\Gamma}},\underline{\mathbf{Y}}}
F_{RQ}(\mathbf{x},\underline{\boldsymbol{\Gamma}},\underline{\mathbf{Y}}).
\end{equation}
Now let us turn to the MM perspective. Let $\mathbf{x}^{(t)}$ be the current iterate. 
For clarity, for each $\omega_m>0$, if $\mathbf{R}_m(\mathbf{x}^{(t)})\succ\mathbf{0}$, set $\boldsymbol{\Gamma}_m^{(t)} = \mathbf{R}_m(\mathbf{x}^{(t)})$ characterized in Theorem~\ref{thm:matrix_rit}.
If instead $\mathbf{R}_m(\mathbf{x}^{(t)})$ is singular, select any maximizing $\boldsymbol{\Gamma}_m^{(t)}$ characterized in~\eqref{eq:gamma_opt_singular}.
Then let $\mathbf{Y}_m^{(t)}$ be the corresponding unique QT maximizer given by~\eqref{eq:Y_operator_update}. 
The matrix SEFP surrogate can be directly defined as
\begin{equation}\label{eq:mm_surrogate}
Q_{\mathrm{SEFP}}(\mathbf{x};\mathbf{x}^{(t)}) \triangleq  F_{RQ} \!\left(\mathbf{x},
\underline{\boldsymbol{\Gamma}}^{(t)},
\underline{\mathbf{Y}}^{(t)}
\right),
\end{equation}
with
\begin{align}
Q_{\mathrm{SEFP}}(\mathbf{x};\mathbf{x}^{(t)})
&\leq
F_R(\mathbf{x},\underline{\boldsymbol{\Gamma}}^{(t)})
\leq F(\mathbf{x}),
\label{eq:mm_lower_bound}
\end{align}
where the first inequality in~\eqref{eq:mm_lower_bound} follows from the matrix QT, and the second follows from the matrix RIT. The equalities hold at the point of $\mathbf{x}^{(t)}$, i.e.,
\begin{align}
Q_{\mathrm{SEFP}}(\mathbf{x}^{(t)};\mathbf{x}^{(t)})
&=
F_R(\mathbf{x}^{(t)},\underline{\boldsymbol{\Gamma}}^{(t)})
=F(\mathbf{x}^{(t)}).
\label{eq:mm_tightness}
\end{align}

\section{INTERCONNECTIONS WITH XMMSE}
\label{sec:xmmse_connection}
It has been observed in the literature that the FP-type and the WMMSE-type algorithms share certain intrinsic interconnections.
Specifically, \cite{FPPart3} first revealed that WMMSE can be interpreted as a specific application of using FP to solve the optimal beamforming problem, and subsequently \cite{BCA+MM WSR} provided a more comprehensive and systematic examination of their intrinsic connections.

Remarkably, such an intrinsic interconnection also exists between our proposed SEFP and the recently proposed XMMSE method~\cite{XMMSE}, essentially aligning with that between FP and WMMSE.
As shown in Fig.~\ref{fig: interconnection_puzzle}, we complete the puzzle of the interconnections between the two FP-type and the two MMSE-type algorithms, where the missing part between SEFP and XMMSE is given in this section.

\begin{figure}[htbp] 
  \centering 
  \hspace*{-0.4cm}\includegraphics[scale=0.76]{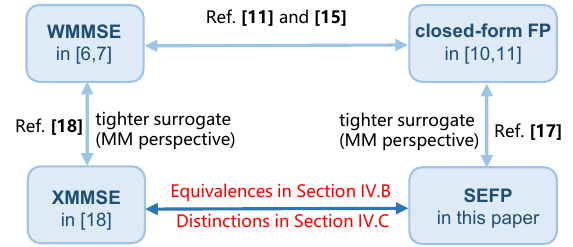} 
  \caption{{Interconnections between the two FP-type and the two MMSE-type algorithms, where the last piece between XMMSE and SEFP is completed.}} 
  \label{fig: interconnection_puzzle} %
\end{figure}

To make this connection explicit, this section proceeds in three steps. First, we recast XMMSE as a functional auxiliary-variable representation, thereby placing XMMSE and SEFP in a common BCA optimization framework.
Second, we establish a lossless finite dimensional factorization of the XMMSE auxiliary-variable path through the SEFP variables and show that, on the optimal auxiliary variable manifold, the two representations induce exactly the same MM surrogate. 
Finally, we clarify the algorithmic implications of this equivalence by showing that XMMSE corresponds to a particular realization within the broader SEFP family, while the additional decoupling freedom of SEFP becomes consequential when scheduling or association is optimized.

\subsection{Functional Auxiliary-Variable Representation of XMMSE}
To place XMMSE and SEFP in the same mathematical framework for comparison, we recast XMMSE directly for the generic weighted sum-log-determinant problem~\eqref{prob:generic_logdet}, without invoking its MSE interpretation.

Specifically, for the $m$-th unweighted term in $F(\mathbf{x})$, and by the variable substitution $\tau=1-\epsilon$,
the exact rate identity underlying XMMSE~\cite[Eq.~(4)]{XMMSE} can be written as
\begin{equation}
\begin{aligned}
\log\det(\mathbf I_{d_m}+\mathbf R_m)
&=\int_0^1
\operatorname{tr}\!\left[(\mathbf I_{d_m}+\tau\mathbf R_m)^{-1}\mathbf R_m\right]
\, d\tau\\
&=\int_0^1
\operatorname{tr}\!\left[\mathbf S_m^H(\mathbf F_m+\tau\mathbf S_m\mathbf S_m^H)^{-1}\mathbf S_m\right]
\, d\tau.
\end{aligned}
\label{eq:xmmse_integral}
\end{equation}
The second equality follows from the push-through identity~\cite[Eq. (167)]{petersen2012matrix}
$$(\mathbf I+\tau\mathbf S^H\mathbf F^{-1}\mathbf S)^{-1}\mathbf S^H\mathbf F^{-1}
=\mathbf S^H(\mathbf F+\tau\mathbf S\mathbf S^H)^{-1}$$ 
when considering $ \mathbf{S}^H$ and $\tau \mathbf{F}^{-1} \mathbf{S}$ as two matrices.

Notice that in~\cite{XMMSE}, the following derivation relied on the physical relationship between the WSR and regularized MMSE receivers. To ensure a concise logical structure, we instead adopt a purely mathematical perspective here. For each $m$ and $\tau\in[0,1]$, introduce \(\mathbf U_{m,\tau}\in\mathbb C^{n_m\times d_m}\). Applying the matrix QT~\cite[Theorem~1]{FPPart3} pointwise yields
\begin{equation}
\begin{aligned}
&\operatorname{tr}\!\left[\mathbf S_m^H(\mathbf F_m+\tau\mathbf S_m\mathbf S_m^H)^{-1}\mathbf S_m\right]=\max_{\mathbf U_{m,\tau}}\Big\{
2\Re\operatorname{tr}(\mathbf S_m^H\mathbf U_{m,\tau})\\
&\hspace{14mm}
-\operatorname{tr}(\mathbf U_{m,\tau}^H\mathbf F_m\mathbf U_{m,\tau})
-\tau \operatorname{tr}(\mathbf U_{m,\tau}^H \mathbf S_m\mathbf S_m^H\mathbf U_{m,\tau})
\Big\},
\end{aligned}
\label{eq:xmmse_inverse_quadratic}
\end{equation}
for which the optimal auxiliary variable is
\begin{equation}
\mathbf U_{m,\tau}^{\star}
=(\mathbf F_m+\tau\mathbf S_m\mathbf S_m^H)^{-1}\mathbf S_m
=\mathbf F_m^{-1}\mathbf S_m(\mathbf I_{d_m}+\tau\mathbf R_m)^{-1}.
\label{eq:xmmse_optimal_receiver}
\end{equation}

Define the uncountable set
$\mathcal U_m\triangleq\{\mathbf U_{m,\tau}:\tau\in[0,1]\}$ as the path with respect to $\tau$ and
$\underline{\mathcal U}\triangleq\{\mathcal U_m\}_{m=1}^{M}$. 
We have the following proposition.

\begin{proposition}[Functional XMMSE Representation] \label{prop:XMMSEobj}
Applying the XMMSE-type transform, the weighted sum-log-determinant problem \eqref{prob:generic_logdet} is equivalent to
\begin{equation}
\begin{aligned}
&\underset{\mathbf{x} \in \mathcal{X},\ \{\mathcal U_m\}_{m=1}^{M}}
{\operatorname{max}}\ F_{X}(\mathbf{x},\underline{\mathcal U})
\triangleq\sum_{m=1}^{M}\omega_m\int_0^1\Big[
2\Re\operatorname{tr}\!\left(\mathbf S_m^H\mathbf U_{m,\tau}\right)\\
&\hspace{12mm}
-\operatorname{tr}(\mathbf U_{m,\tau}^H\mathbf F_m\mathbf U_{m,\tau})
-\tau \operatorname{tr}(\mathbf U_{m,\tau}^H \mathbf S_m\mathbf S_m^H\mathbf U_{m,\tau})
\Big] \, d\tau
\end{aligned}
\label{eq:xmmse_functional_objective}
\end{equation}
where the optimal XMMSE-induced auxiliary variable path $\mathcal U_m^\star$ is given pointwise by
\eqref{eq:xmmse_optimal_receiver}.
\end{proposition}

\begin{proof}
Combining \eqref{eq:xmmse_integral} and \eqref{eq:xmmse_inverse_quadratic} gives
\eqref{eq:xmmse_functional_objective}. 
Since $\mathbf{F}+\tau \mathbf{S}\mathbf{S}^H\succ \mathbf 0$ for every $\tau\in[0,1]$, the maximizer in \eqref{eq:xmmse_inverse_quadratic} is unique and equals \eqref{eq:xmmse_optimal_receiver}.
\end{proof}

Note that the XMMSE-type objective $F_{ X}(\mathbf{x},\underline{\mathcal U})$ in~\eqref{eq:xmmse_functional_objective}
is precisely the counterpart of the sum-regularized-MSE formulation in~\cite[Eq. (3)]{XMMSE}. That is,
\(F_{ X}(\mathbf{x},\underline{\mathcal U})\) plays for XMMSE the same representational role as $F_{RQ}(\mathbf{x},\underline{\boldsymbol{\Gamma}},\underline{\mathbf{Y}})$ does for SEFP--- maximizing the corresponding auxiliary variables recovers the same original objective $F(\mathbf{x})$.

\subsection{Two Coordinates Mapping Under the Optimal Auxiliary Variable Manifold}
In this subsection, we show that the auxiliary variables for SEFP and XMMSE have a similar yet subtler relationship than those for FP and WMMSE.
Specifically, the two algorithms are essentially two coordinate representations of the same optimization object when restricted to their optimal auxiliary-variable manifold, analogous to the FP--WMMSE correspondence in~\cite{BCA+MM WSR}.

Since a term with $\omega_m=0$ does not affect the original objective, we restrict the discussion in this subsection to $\omega_m>0$ without loss of generality.  
For each $m$ and any $\boldsymbol{\Gamma}_m\in\mathbb{H}_{++}^{d_m}$, define the $\tau$-dependent ($\tau \in [0,1]$) spectral weighting matrix
\begin{equation}
\mathbf K_{m,\tau}(\boldsymbol{\Gamma}_m)
\triangleq
\boldsymbol{\Gamma}_m(\mathbf I_{d_m}+\tau\boldsymbol{\Gamma}_m)^{-1}
(\mathbf I_{d_m}+\boldsymbol{\Gamma}_m)^{1/2}.
\label{eq:xmmse_kernel}
\end{equation}
Under the spectral functional calculus defined in Section~\ref{subsec:Hermitian functional calculus}, all factors in~\eqref{eq:xmmse_kernel} are functions of the same Hermitian matrix $\boldsymbol{\Gamma}_m$, hence they are simultaneously
diagonalizable by the eigenvectors of $\boldsymbol{\Gamma}_m$ and therefore commute.

\begin{lemma}[Spectral Moment Identities]\label{lem:kernel_moments}
For every $\boldsymbol{\Gamma}_m\in\mathbb{H}_{++}^{d_m}$, the spectral weighting matrix in~\eqref{eq:xmmse_kernel} satisfies
\begin{subequations}
\label{eq:kernel_moment_identities}
\begin{align}
\int_0^1\mathbf K_{m,\tau}(\boldsymbol{\Gamma}_m)\, d\tau
=c(\boldsymbol{\Gamma}_m)&^{1/2},
\label{eq:kernel_first_moment}\\
\int_0^1\mathbf K_{m,\tau}(\boldsymbol{\Gamma}_m)
\mathbf K_{m,\tau}(\boldsymbol{\Gamma}_m)^H\, d\tau
&=\boldsymbol{\Gamma}_m^2,
\label{eq:kernel_second_moment}\\
\int_0^1\tau\mathbf K_{m,\tau}(\boldsymbol{\Gamma}_m)
\mathbf K_{m,\tau}(\boldsymbol{\Gamma}_m)^H\, d\tau
&=b(\boldsymbol{\Gamma}_m).
\label{eq:kernel_weighted_second_moment}
\end{align}
\end{subequations}
\end{lemma}

\begin{proof}
Let
$\boldsymbol{\Gamma}_m =\boldsymbol{\Xi}_m
\operatorname{diag}(\gamma_{m,1},\ldots,\gamma_{m,d_m})
\boldsymbol{\Xi}_m^H$.
Then $\mathbf K_{m,\tau}(\boldsymbol{\Gamma}_m)$ has eigenvalues
$\{\gamma_{m,i}\sqrt{1+\gamma_{m,i}}/(1+\tau\gamma_{m,i})\}_{i=1}^{d_m}$.
The three identities in~\eqref{eq:kernel_moment_identities} follow by integrating this scalar function itself,
its squared magnitude, and its $\tau$-weighted squared magnitude over $[0,1]$, respectively.
Applying the resulting scalar identities eigenvalue-wise according to the matrix function property in Section~\ref{subsec:Hermitian functional calculus} then yields the corresponding matrix identities.
\end{proof}

For any $\boldsymbol{\Gamma}_m\in\mathbb{H}_{++}^{d_m}$ and $\mathbf Y_m\in\mathbb{C}^{n_m\times d_m}$, define the structured XMMSE-induced auxiliary variable
\begin{equation}
\mathbf U_{m,\tau}(\boldsymbol{\Gamma}_m,\mathbf Y_m)
\triangleq
\frac{1}{\sqrt{\omega_m}}\,
\mathbf Y_m\mathbf K_{m,\tau}(\boldsymbol{\Gamma}_m),
\label{eq:structured_receiver_path}
\end{equation}
for $\tau\in[0,1]$ and let
$\mathcal U_m(\boldsymbol{\Gamma}_m,\mathbf Y_m)
\triangleq\{\mathbf U_{m,\tau}(\boldsymbol{\Gamma}_m,\mathbf Y_m):\tau\in[0,1]\}$ be the path with respect to $\tau$.
Collecting these paths gives
$\underline{\mathcal U}(\underline{\boldsymbol{\Gamma}},\underline{\mathbf Y})
\triangleq\{\mathcal U_m(\boldsymbol{\Gamma}_m,\mathbf Y_m)\}_{m=1}^{M}$.

\begin{theorem}[Lossless XMMSE--SEFP Factorization]\label{SEFP-XMMSE}
For every feasible $\mathbf x\in\mathcal X$,
$\boldsymbol{\Gamma}_m\in\mathbb{H}_{++}^{d_m}$ and
$\mathbf Y_m\in\mathbb{C}^{n_m\times d_m}$, it follows
\begin{equation}
F_X\!\left(
\mathbf x,
\underline{\mathcal U}(\underline{\boldsymbol{\Gamma}},\underline{\mathbf Y})
\right)
=
F_{RQ}(\mathbf x,\underline{\boldsymbol{\Gamma}},\underline{\mathbf Y}).
\label{eq:xmmse_sefp_objective_identity}
\end{equation}
Moreover, suppose $\mathbf R_m\succ\mathbf 0$ for all $m$.
By Theorem~\ref{thm:matrix_sefp}, the optimal SEFP-induced auxiliary variables
\begin{subequations}
\label{eq:xmmse_sefp_optimal_mapping}
\begin{align}
\boldsymbol{\Gamma}_m^\star
&=\mathbf R_m
=\mathbf S_m^H\mathbf F_m^{-1}\mathbf S_m,
\label{eq:xmmse_sefp_gamma_mapping}\\
\mathbf Y_m^\star
&=\sqrt{\omega_m}\,
\mathbf F_m^{-1}\mathbf S_m\mathbf R_m^{-1}
(\mathbf I+\mathbf R_m)^{-1/2}
\label{eq:xmmse_sefp_y_mapping}
\end{align}
\end{subequations}
recover the unique optimal XMMSE-induced auxiliary variables in~\eqref{eq:xmmse_optimal_receiver} through~\eqref{eq:structured_receiver_path} and the converse holds as well.  
Consequently,
\begin{align}
\max_{\underline{\mathcal U}}
F_X(\mathbf x,\underline{\mathcal U})
&=
\max_{\underline{\boldsymbol{\Gamma}},\underline{\mathbf Y}}
F_X\!\left(
\mathbf x,
\underline{\mathcal U}(\underline{\boldsymbol{\Gamma}},\underline{\mathbf Y})
\right)
\nonumber\\
&=
\max_{\underline{\boldsymbol{\Gamma}},\underline{\mathbf Y}}
F_{RQ}(\mathbf x,\underline{\boldsymbol{\Gamma}},\underline{\mathbf Y})
=F(\mathbf x).
\label{eq:xmmse_sefp_lossless}
\end{align}
\end{theorem}

\begin{proof} See Appendix~\ref{proof_of_thm_SEFP-XMMSE}.
\end{proof}

\begin{remark}[Rank deficient case]\label{rem:xmmse_sefp_rank_deficient}
If $\mathbf{R}_m$ is singular and has the eigendecomposition in~\eqref{eq:Rm_decom}.
Select any optimal $\boldsymbol{\Gamma}_m^\star$ of the form in~\eqref{eq:gamma_opt_singular} and let
\begin{equation}
\mathbf{Y}_m^\star = \sqrt{\omega_m}\, \mathbf{F}_m^{-1}\mathbf{S}_m\mathbf{R}_m^\dagger(\mathbf{I}+\mathbf{R}_m)^{-1/2}.
\label{eq:rank_deficient_xmmse_y}
\end{equation}
Since $\mathbf{S}_m$ vanishes on $\ker(\mathbf{R}_m)$,
\eqref{eq:structured_receiver_path} still recovers
\eqref{eq:xmmse_optimal_receiver}.
\end{remark}

Theorem~\ref{SEFP-XMMSE} shows the SEFP-induced auxiliary variables could be seen as a lossless finite-dimensional factorization of the XMMSE-induced auxiliary variable path, or rather, the two can be regarded as two representations in different coordinates and are interchangeable at optimality. 

Next, we lift the above coordinate-level equivalence to the algorithmic level under the MM theory.  
Let $\mathbf{x}^{(t)}$ be the current iterate and denote
$\underline{\mathcal U}^{(t)}
\triangleq
\argmax_{\underline{\mathcal U}}
F_X(\mathbf x^{(t)},\underline{\mathcal U}),$ where the maximizer is uniquely specified pointwise by~\eqref{eq:xmmse_optimal_receiver}. Define the XMMSE surrogate as
\begin{equation}
Q_X(\mathbf x;\mathbf x^{(t)})
\triangleq
F_X(\mathbf x,\underline{\mathcal U}^{(t)}).
\label{eq:xmmse_surrogate}
\end{equation}
Recalling that $Q_{\mathrm{SEFP}}(\mathbf x;\mathbf x^{(t)})$ is defined in~\eqref{eq:mm_surrogate} by fixing the SEFP auxiliary variables at their exact maximizers associated with the same $\mathbf x^{(t)}$, we have the following proposition.

\begin{proposition}[Identical Updates of Surrogates and Original Variables] \label{coro:Equality_of_surro}
For every $x\in\mathcal X$, we have
\begin{equation}
Q_{X}(\mathbf{x};\mathbf{x}^{(t)}) = F_{RQ} \!\left(\mathbf{x},
\underline{\boldsymbol{\Gamma}}^{(t)},
\underline{\mathbf{Y}}^{(t)}
\right) = Q_{\mathrm{SEFP}}(\mathbf{x};\mathbf{x}^{(t)}).
\label{eq:xmmse_sefp_surrogate_equality}
\end{equation}
Furthermore, suppose that XMMSE and matrix SEFP initialize with the same feasible original variable $\mathbf x^{(0)}$, update their auxiliary variables exactly at every iteration, solve the induced surrogate subproblems exactly, and select the same maximizer whenever that subproblem has multiple solutions.  
Then their original optimization variable updates are identical as 
\begin{equation}
\mathbf x_X^{(t)}
=
\mathbf x_{\mathrm{SEFP}}^{(t)},
\qquad t=0,1,\ldots,
\label{eq:xmmse_sefp_identical_iterates}
\end{equation}
and consequently
$F(\mathbf x_X^{(t)})=F(\mathbf x_{\mathrm{SEFP}}^{(t)})$ for every $t$.
\end{proposition}
\begin{proof}
By Theorem~\ref{SEFP-XMMSE} and Remark~\ref{rem:xmmse_sefp_rank_deficient}, the exact SEFP maximizers $(\underline{\boldsymbol{\Gamma}}^{(t)},\underline{\mathbf Y}^{(t)})$ at $\mathbf x^{(t)}$ generate the unique XMMSE maximizing path $\underline{\mathcal U}^{(t)}$ through~\eqref{eq:structured_receiver_path}.  
Once these auxiliary variables are fixed,~\eqref{eq:xmmse_sefp_objective_identity} holds for every candidate $\mathbf x\in\mathcal X$, which proves~\eqref{eq:xmmse_sefp_surrogate_equality}.

Equation~\eqref{eq:xmmse_sefp_identical_iterates} follows by induction.  
If $\mathbf x_X^{(t)}=\mathbf x_{\mathrm{SEFP}}^{(t)}$, ~\eqref{eq:xmmse_sefp_surrogate_equality} shows that the two methods construct exactly the same surrogate over the same feasible set $\mathcal X$.
Exact maximization with the same selection rule therefore gives $\mathbf x_X^{(t+1)}=\mathbf x_{\mathrm{SEFP}}^{(t+1)}$, completing the induction.
\end{proof}
\begin{remark}
As illustrated in Fig.~\ref{fig: interconnection_puzzle}, Proposition~\ref{coro:Equality_of_surro} completes the puzzle of the interconnections between two FP-type and two MMSE-type algorithms, showing that those two lines of reasoning strategies achieve equivalent surrogate enhancement from the surrogate construction perspective.

On one hand, the improvement of XMMSE over WMMSE lies primarily in how the original WSR maximization problem is transformed into an MSE-type minimization problem.
Whereas WMMSE employs a variational reformulation of the rate--MSE relation by introducing auxiliary MSE weight matrices, which are iteratively updated together with the receive and transmit beamformers,
XMMSE exploits an exact integral rate--MSE identity to obtain an equivalent sum-regularized-MSE reformulation without introducing explicit MSE weight matrices.

On the other hand, the SEFP originates from revisiting the LDT in classical closed-form FP from a surrogate construction perspective.
The construction of the LDT-induced lower bound couples the constant term and the variable ratio term inside the logarithmic function inherently and consequently leads to a relatively conservative surrogate.
In contrast, the RIT reconstructs the lower bound in a reciprocal-inversion coordinate, which separates the above-mentioned structural coupling.
\end{remark}

\begin{remark} \label{euqi_bound}
For wireless communications, the equivalence between the surrogate and the original variables established in Proposition~\ref{coro:Equality_of_surro} applies to a broad class of network-utility-maximization problems, including SISO power control~\cite{FPPart1} and MIMO beamforming in fixed single association D2D networks~\cite{FPPart3}, and multi-stream MIMO precoding in single-cell~\cite{WMMSE_1} and multicell MU-MIMO networks~\cite{XMMSE,WMMSE_2}.
A common structural feature of these problems is that the set of utility terms and their associated weights are fixed.
The optimization variables modify the signal and interference quantities entering each utility term, but neither the identity of the term nor its weight changes. 
\end{remark}

However, such a structure changes once scheduling or association decisions are incorporated as optimization variables.
In such problems, the active utility term and its associated weight can themselves depend on the scheduling decisions, in which case SEFP retains additional flexibility over XMMSE, as will be elaborated analytically next.

\subsection{XMMSE as a Special SEFP Instance: From Continuous to Mixed-Discrete Optimization}
\label{subsec:xmmse_special_case}
Remarkably, SEFP inherits from FP two flexibilities  -- nonunique ratio decouplings and BCA block-selection rules -- to generate distinct algorithmic instances~\cite{BCA+MM WSR}. By these two flexibilities, we have the following proposition.

\begin{proposition}
XMMSE is one particular member of the broad SEFP algorithm family under a specific ratio decoupling and a specific BCA order, analogous to the WMMSE-FP relationship.
\end{proposition}
\begin{proof}
The argument is due to two flexibilities inherited from FP: nonunique ratio decouplings and BCA block-selection.

To make the first flexibility explicit, we keep the RIT unchanged and revisit the $m$-th weighted term of the RIT-induced objective $F_{R}$ in~\eqref{eq:operator_kron_equivalence}. 
By the matrix function property in Section~\ref{subsec:Hermitian functional calculus}, $c(\boldsymbol{\Gamma}_m)^{1/2}$, $\boldsymbol{\Gamma}_m^2$, and $b(\boldsymbol{\Gamma}_m)$ are all functions of the same Hermitian matrix $\boldsymbol{\Gamma}_m$ and therefore commute with one another.
Consequently,
\begin{equation}
\mathcal D_{m,\boldsymbol{\Gamma}_m}^{-1}
\big(
\mathbf S_mc(\boldsymbol{\Gamma}_m)^{1/2}
\big)
=
\big(\mathcal D_{m,\boldsymbol{\Gamma}_m}^{-1}(\mathbf S_m)\big)c(\boldsymbol{\Gamma}_m)^{1/2}.
\label{eq:sefp_operator_bridge_c}
\end{equation}
Using the trace cyclic invariance property further gives
\begin{equation}
\begin{aligned}
\Big\langle&
\mathbf S_m c(\boldsymbol{\Gamma}_m)^{1/2},
\mathcal D_{m,\boldsymbol{\Gamma}_m}^{-1}
\big(
\mathbf S_m c(\boldsymbol{\Gamma}_m)^{1/2}
\big)
\Big\rangle_F
\\& \qquad =
\Big\langle
\mathbf S_m c(\boldsymbol{\Gamma}_m),
\mathcal D_{m,\boldsymbol{\Gamma}_m}^{-1}(\mathbf S_m)
\Big\rangle_F,
\label{eq:sefp_operator_bridge_d}
\end{aligned}
\end{equation}
which shows that the same RIT-induced ratio admits another exact QT decoupling.
Instead of placing $c(\boldsymbol{\Gamma}_m)^{1/2}$ symmetrically in the QT numerator, one may merge the two factors into $c(\boldsymbol{\Gamma}_m)$.
By introducing the distinct QT-induced auxiliary matrix $\mathbf Y_m^{\prime}$, the corresponding SEFP objective becomes~\eqref{eq:alt_sefp_decoupling}.
\begin{figure*}[t]
\begin{equation}
 F^{\prime}_{RQ}(\mathbf{x},\underline{\boldsymbol{\Gamma}},\underline{\mathbf{Y}^{\prime}}) \triangleq
\sum_{m=1}^{M} \omega_m
\Big\{2\Re\tr\!\big[c(\boldsymbol{\Gamma}_m)\mathbf S_m^H\mathbf Y_m^{\prime}\big]-\tr\!\big[c(\boldsymbol{\Gamma}_m)\mathbf Y_m^{\prime H}\mathbf F_m\mathbf Y_m^{\prime}\boldsymbol{\Gamma}_m^2\big]
-\tr\!\big[c(\boldsymbol{\Gamma}_m)\mathbf Y_m^{\prime H}\mathbf S_m\mathbf S_m^H\mathbf Y_m^{\prime} b(\boldsymbol{\Gamma}_m)\big]\Big\}.
\label{eq:alt_sefp_decoupling}
\end{equation}
\hrulefill 
\vspace*{-10pt} 
\end{figure*}
Notice that, for every fixed $\boldsymbol{\Gamma}_m\succ\mathbf0$, the coordinate change
\begin{equation}
\mathbf Y_m=\sqrt{\omega_m}\,\mathbf Y_m^{\prime} c(\boldsymbol{\Gamma}_m)^{1/2}
\label{eq:alt_sefp_coordinate_map}
\end{equation}
is bijective and transforms \eqref{eq:alt_sefp_decoupling} exactly into~\eqref{def_of_FmSEFP}. 
Therefore, both formulations maximize the same RIT-induced quantity for fixed $\boldsymbol{\Gamma}_m$, but they correspond to different coordinate realizations of the QT block.

The second flexibility is the BCA update rule.
Define the three auxiliary-block maximization maps
\begin{subequations}\label{eq:sefp_block_maps}
\begin{align}
\mathcal T_R(\mathbf x)
&\triangleq\argmax_{\underline{\boldsymbol{\Gamma}}}F_R(\mathbf x,\underline{\boldsymbol{\Gamma}}),
\label{eq:block_map_rit}\\
\mathcal T_{\Gamma}(\mathbf x,\underline{\mathbf Y})
&\triangleq\argmax_{\underline{\boldsymbol{\Gamma}}}F_{RQ}(\mathbf x,\underline{\boldsymbol{\Gamma}},\underline{\mathbf Y}),
\label{eq:block_map_gamma}\\
\mathcal T_Y(\mathbf x,\underline{\boldsymbol{\Gamma}})
&\triangleq\argmax_{\underline{\mathbf Y}}F_{RQ}(\mathbf x,\underline{\boldsymbol{\Gamma}},\underline{\mathbf Y}).
\label{eq:block_map_y}
\end{align}
\end{subequations}
At the same $\mathbf x^{(t)}$, considering one entire iteration, different admissible update orders therefore generate different SEFP variants, e.g.,
\begin{subequations}\label{eq:sefp_bca_variants}
\begin{align}
\mathcal A_X:\quad
&\underline{\boldsymbol{\Gamma}}^{(t+1)}\in\mathcal T_R(\mathbf x^{(t)}),\
\underline{\mathbf Y}^{(t+1)}\in\mathcal T_Y(\mathbf x^{(t)},\underline{\boldsymbol{\Gamma}}^{(t+1)}),
\label{eq:sefp_bca_x}\\
\mathcal A_{\Gamma Y}:\
&\underline{\boldsymbol{\Gamma}}^{(t+1)}\in\mathcal T_{\Gamma}(\mathbf x^{(t)},\underline{\mathbf Y}^{(t)}), \
\underline{\mathbf Y}^{(t+1)}\in\mathcal T_Y(\mathbf x^{(t)},\underline{\boldsymbol{\Gamma}}^{(t+1)}),
\label{eq:sefp_bca_gy}\\
\mathcal A_{Y\Gamma}:\
&\underline{\mathbf Y}^{(t+1)}\in\mathcal T_Y(\mathbf x^{(t)},\underline{\boldsymbol{\Gamma}}^{(t)}), \
\underline{\boldsymbol{\Gamma}}^{(t+1)}\in\mathcal T_{\Gamma}(\mathbf x^{(t)},\underline{\mathbf Y}^{(t+1)}).
\label{eq:sefp_bca_yg}
\end{align}
\end{subequations}
Each sequence is followed by the common original variable update
$\mathbf x^{(t+1)}\in\argmax_{\mathbf x\in\mathcal X}F_{RQ}(\mathbf x,\underline{\boldsymbol{\Gamma}}^{(t+1)},\underline{\mathbf Y}^{(t+1)})$.
The mixed-level rule $\mathcal A_X$ first maximizes the RIT objective and then the QT block; by Theorem~\ref{SEFP-XMMSE}, Proposition~\ref{coro:Equality_of_surro}, and \eqref{eq:structured_receiver_path}, it reproduces XMMSE.
In contrast, $\mathcal A_{\Gamma Y}$ and $\mathcal A_{Y\Gamma}$ are conventional cyclic BCA variants of the single transformed objective $F_{RQ}$ and generally need not coincide with $\mathcal A_X$.
The above analyses complete the proof.
\end{proof}

The above flexibilities become even more consequential when discrete variables (e.g., scheduling) are part of the optimization. 
In what follows, we present a conceptual explanation, leaving the detailed analysis to Section~\ref{SEFPLinQ}. 
The essential change introduced by scheduling can be represented by allowing a fixed utility weight to become scheduling-dependent, i.e., $\omega_j\rightarrow\omega_{j,s_j}$, so that the active utility term and its weight vary with the discrete decision. 

By applying the SEFP-type transform to the original objective with an appropriate ratio decoupling, the relevant part to the original variables is then exposed in an edge-separable form and can be optimized independently.
In contrast, the XMMSE type representation associates each candidate link with a continuum of auxiliary variables, 
the candidate identity and the network interference are therefore coupled that prevent the edge-wise separation, yielding a markedly less convenient combinatorial structure. 
A more practical and instructive derivation of this distinction, after the scheduling variables have been explicitly introduced, is given in the next section.


\section{JOINT SCHEDULING, AND BEAMFORMING VIA MATRIX SEFP}\label{SEFPLinQ}
To establish a direct comparison with the FPLinQ strategy~\cite[Algorithm~2]{FPPart3} and demonstrate the structural advantage over XMMSE identified in Section~\ref{subsec:xmmse_special_case}, this section applies the matrix SEFP framework developed in Section~\ref{sec:matrix_rit} to the joint link scheduling and beamforming WSR optimization problem in a flexible--association MIMO interference network, under the same setup as in~\cite[Section~V]{FPPart3}.

\subsection{System Model and Problem Formulation}
\label{subsec:system_model}
Consider a typical multiple-antenna flexible--association D2D wireless interference network consisting of a transmitter set $\mathcal I$ and a receiver set $\mathcal J$. 
Each transmitter is equipped with $N_t$ antennas, each receiver is equipped with $N_r$ antennas, and each active link transmits $d$ independent data streams for simplicity.
Let $\mathcal K_j\subseteq\mathcal I$ denote the set of transmitters that may serve receiver $j$ and $\mathcal L_i\subseteq\mathcal J$ denote the set of receivers that may be served by transmitter $i$.
The scheduling variable associated with receiver $j$ is denoted by $s_j\in\mathcal K_j\cup\{\varnothing\}$, for
$s_j=i$ means transmitter $i$ is scheduled to serve receiver $j$, while $s_j=\varnothing$ means receiver $j$ is inactive.

Defining the binary association variable
\begin{equation}
q_{ji} \triangleq 1 \text{ if } s_j=i, \text{ else } 0
\end{equation}
and, correspondingly, $\underline{\mathbf{q} }\triangleq {\{q_{ji}|\ i\in\mathcal I,j\in\mathcal J\}}$.
In each timeslot, each transmitter (or receiver) can only communicate with at most one of its associated receivers (or transmitters), respectively. Mathematically, a legal $\underline{\mathbf{q}}$ must satisfy 
\begin{equation}\label{eq:matching_constraints}
\sum_{i\in\mathcal K_j}q_{ji}\leq 1,\
 \forall j\in\mathcal J,\qquad
 \sum_{j\in\mathcal L_i}q_{ji}\leq 1,\ \forall i\in\mathcal I.
\end{equation}

Let $\mathbf H_{ji}\in\C^{N_r\times N_t}$ denote the channel from transmitter $i$ to receiver $j$, and let $\mathbf V_i\in\C^{N_t\times d}$ denote the transmit beamforming matrix of transmitter $i$ supporting $d$ datastreams transmission.
The transmitted data vector $\mathbf x_i$ is assumed to satisfy $\mathbb E[\mathbf x_i\mathbf x_i^{\mathrm H}]=\mathbf I_d$.

For an active receiver $j$ with $s_j=i$, define the effective desired-signal matrix $ \mathbf S_j\triangleq \sum_{i \in \mathcal{K}_j} q_{ji} \mathbf{H}_{ji} \mathbf{V}_i,$
the interference-plus-noise covariance matrix
$\mathbf F_j \triangleq \sigma_j^2 \mathbf I_{N_r} + \sum_{\substack{k \in \mathcal{J} \\ k \ne j}} \sum_{i \in \mathcal{K}_k} q_{ki} \mathbf H_{ji} \mathbf V_i \mathbf V_i^{H} \mathbf H_{ji}^{H},$ and correspondingly, the SINR matrix
$\mathbf R_j\triangleq
\mathbf S_j^{ H}\mathbf F_j^{-1}\mathbf S_j\succeq\mathbf 0.$
Under Gaussian signaling and treating interference as noise, the achievable rate of receiver $j$ is
$
R_j(\underline{\mathbf{s}},\underline{\mathbf{V}}) = \log\det\!\left( \mathbf I_d+\mathbf R_j \right). $

Given $\omega_{ji}\ge0$ for each associated link from transmitter $i$ to receiver $j$, and $\omega_{j,\varnothing} \triangleq 0$, the network joint scheduling and beamforming WSR optimization problem is given by
\begin{subequations}\label{prob:wsr}
\begin{align}
\mathop{\mathrm{max}}_{\underline{\mathbf{s}},\underline{\mathbf{V}}}\quad
& f(\underline{\mathbf{s}},\underline{\mathbf{V}}) \triangleq \sum_{j\in\mathcal{J}} \omega_{j,s_j} R_j(\underline{\mathbf{s}},\underline{\mathbf{V}}) \label{prob:wsr_obj}\\
\mathrm{s.t.}\quad
& \tr\left(\mathbf{V}_i^{H}\mathbf{V}_i\right) \le P_i, \quad \forall i\in\mathcal{I}, \label{prob:wsr_power}\\
& s_j \in \mathcal{K}_j \cup \{\varnothing\}, \quad \forall j\in\mathcal{J}, \label{prob:wsr_schedule}\\
& s_j \neq s_k, \text{ or } s_j= \varnothing, \forall j \neq k , \label{prob:wsr_one_to_one}
\end{align}
\end{subequations}
where $\underline{\mathbf{s}}$ and $\underline{\mathbf{V}}$ denote the scheduling links set and the transmitter beamformers set, respectively, and~\eqref{prob:wsr_power} is the transmit power constraint.
Problem~\eqref{prob:wsr} is a mixed discrete--continuous nonconvex optimization problem that is NP-hard and notoriously difficult to solve.

\subsection{Iterative Optimization via Matrix SEFP -- SEFPLinQ}
\label{subsec:iterative_sefp}

\subsubsection{Matrix SEFP reformulation and the scaling version}
By applying the matrix SEFP, specializing the generic variable $\mathbf{x}$ in~\eqref{def_of_FmSEFP} by $(\underline{\mathbf{s}},\underline{\mathbf{V}})$, problem~\eqref{prob:wsr} is equivalent to
\begin{subequations}\label{prob:sefp_equiv}
\begin{align}
&\underset{\underline{\mathbf s},\underline{\mathbf V},\underline{\boldsymbol\Gamma},\underline{\mathbf Y}}{\max}\quad
f_{\mathrm{SE}}(\underline{\mathbf s},\underline{\mathbf V},\underline{\boldsymbol\Gamma},\underline{\mathbf Y})
\label{f_SE}\\
& \quad\text{s.t.}\quad  (\text{\ref{prob:wsr_power})--(\ref{prob:wsr_one_to_one}}),\ \
\boldsymbol\Gamma_j\in\mathbb{H}_{++}^{d},
\ \mathbf Y_j\in\C^{N_r\times d},
\ \forall j,
\label{prob:sefp_equiv_aux}
\end{align}
\end{subequations}
where the objective $f_{\mathrm{SE}}(\mathbf s,\mathbf V,\boldsymbol\Gamma,\mathbf Y)$ is defined in~\eqref{eq:sefp_objective} on the top of this page.
\begin{figure*}[!t]
\normalsize
\begin{equation} \label{eq:sefp_objective}
  f_{\mathrm{SE}}(\underline{\mathbf s},\underline{\mathbf V},\underline{\boldsymbol\Gamma},\underline{\mathbf Y})
\triangleq \sum_{j\in\mathcal J}\Big\{
2\sqrt{\omega_{j,s_j}}\,
\Re \tr\!\left[
 c(\boldsymbol\Gamma_j)^{1/2}
 \mathbf S_j^{H} \mathbf Y_j \right]
 -\tr\!\left(\mathbf Y_j^{\mathrm H}\mathbf F_j\mathbf Y_j \boldsymbol\Gamma_j^2\right)
 -\tr\!\left[ \mathbf Y_j^{H} \mathbf S_j \mathbf S_j^{H} \mathbf Y_j b(\boldsymbol\Gamma_j)\right]  \Big\}.
\end{equation}
\hrulefill 
\vspace*{4pt} 
\end{figure*}

Although~\eqref{eq:sefp_objective} is the exact consequence of the matrix SEFP transform, in practice, we choose a more numerically implementable form to move forward.
The idea is quite natural because without this operation, to deal with the case where $\mathbf{R}_j$ is rank deficient, the optimal $\boldsymbol\Gamma_j$ has to be updated extra by~\eqref{eq:gamma_opt_singular}.
Specifically, we first introduce two scalar functions
$\rho(\cdot)$ and $\delta(\cdot)$ on $\mathbb R_{+}$ as
\begin{subequations}\label{eq:rho_delta_scalar}
\begin{align}
\rho(t)
&\triangleq
\begin{cases}
\dfrac{\sqrt{1+t}\,\log(1+t)}{t}, & t>0,\\[1.4mm]
\lim_{t\downarrow 0}\rho(t)=1, & t=0,
\end{cases}
\label{eq:rho_scalar}\\
\delta(t)
&\triangleq
\begin{cases}
\dfrac{(1+t)\log(1+t)-t}{t^2}, & t>0,\\[1.4mm]
\lim_{t\downarrow 0}\delta(t)=\dfrac{1}{2}, & t=0,
\end{cases}
\label{eq:delta_scalar}
\end{align}
\end{subequations}
where their both values at $t=0$ are defined by continuity.

We next extend $\rho(\cdot)$ and $\delta(\cdot)$ to Hermitian matrix functions. 
Specifically, for any $\boldsymbol\Gamma$ with the form in~\eqref{eq:gamma_evd}, define
\begin{subequations}\label{eq:rho_delta_matrix}
\begin{align}
\rho(\boldsymbol\Gamma)
&\triangleq
\boldsymbol{\Xi}
\operatorname{diag}\!\big(
\rho(\gamma_1),\ldots,\rho(\gamma_d)
\big)
\boldsymbol{\Xi}^{H},
\label{eq:rho_matrix}\\
\delta(\boldsymbol\Gamma)
&\triangleq
\boldsymbol{\Xi}
\operatorname{diag}\!\big(
\delta(\gamma_1),\ldots,\delta(\gamma_d)
\big)
\boldsymbol{\Xi}^{H}.
\label{eq:delta_matrix}
\end{align}
\end{subequations}
Accordingly, zero eigenvalues are handled through
$\rho(0)=1$ and $\delta(0)=1/2$, so
\eqref{eq:rho_delta_matrix} remains well defined even when $\boldsymbol\Gamma$ is rank deficient. 
For $\boldsymbol\Gamma\succ\mathbf0$, these spectral definitions reduce to
\begin{equation}
\rho(\boldsymbol\Gamma)
=
\boldsymbol\Gamma^{-1}
c(\boldsymbol\Gamma)^{1/2},
\qquad
\delta(\boldsymbol\Gamma)
=
\boldsymbol\Gamma^{-2}
b(\boldsymbol\Gamma).
\label{eq:rho_delta_pd_equiv}
\end{equation}

We next introduce the scaled auxiliary variable $\mathbf Z_j\triangleq
\mathbf Y_j\boldsymbol\Gamma_j.$
For $\boldsymbol\Gamma_j\succ\mathbf0$,
$\mathbf Y_j\leftrightarrow\mathbf Z_j$ is bijective. 
Substituting
$\mathbf Y_j=\mathbf Z_j\boldsymbol\Gamma_j^{-1}$ into~\eqref{eq:sefp_objective}, and using the cyclic invariance property of the trace yields the scaling objective
\begin{align}
&\bar f_{\mathrm{SE}}\!\left(
\underline{\mathbf s},\underline{\mathbf V},
\underline{\boldsymbol\Gamma},\underline{\mathbf Z}
\right)
\triangleq\sum_{j\in\mathcal J}\Big\{
2\sqrt{\omega_{j,s_j}}\,
\Re\tr\!\left[ \rho(\boldsymbol\Gamma_j) \mathbf S_j^{H}\mathbf Z_j  \right] \nonumber
\\
& \qquad \qquad-\tr\!\left(
 \mathbf Z_j^{\mathrm H}\mathbf F_j\mathbf Z_j
\right)
-\tr\!\left[
 \mathbf Z_j^{\mathrm H}\mathbf S_j\mathbf S_j^{\mathrm H}
 \mathbf Z_j\delta(\boldsymbol\Gamma_j)
\right]
\Big\}.
\label{eq:scaled_sefp_objective}
\end{align}
In particular, for every $\boldsymbol\Gamma_j\succ\mathbf0$ and
$\mathbf Z_j=\mathbf Y_j\boldsymbol\Gamma_j$, we have
\begin{equation}
f_{\mathrm{SE}}\!\left(
\underline{\mathbf s},\underline{\mathbf V},
\underline{\boldsymbol\Gamma},\underline{\mathbf Y}
\right)
=
\bar f_{\mathrm{SE}}\!\left(
\underline{\mathbf s},\underline{\mathbf V},
\underline{\boldsymbol\Gamma},\underline{\mathbf Z}
\right).
\label{eq:scaled_unscaled_identity}
\end{equation}
Thus, the scaling is an exact reparameterization and does not alter either the objective value or the optimization variables $\underline{\mathbf s}$ and $\underline{\mathbf V}$. 
Moreover, since the scalar functions $\rho(t)$ and $\delta(t)$ are continuous
on $\mathbb R_{+}$, their spectral extensions in
\eqref{eq:rho_delta_matrix} are well defined on the entire cone
$\mathbb H_{+}^{d}$ and remain continuous when one or more
eigenvalues of $\boldsymbol\Gamma$ approach zero.
Consequently, although the change of variables
$\mathbf Z_j=\mathbf Y_j\boldsymbol\Gamma_j$ is bijective only for
$\boldsymbol\Gamma_j\succ\mathbf0$, the scaled objective obtained
below admits a natural continuous extension of $\boldsymbol\Gamma_j$ from
$\mathbb H_{++}^{d}$ to $\mathbb H_{+}^{d}$.
This extension enables the auxiliary variable update $
\boldsymbol\Gamma_j^\star=\mathbf R_j$
to be applied uniformly, including the case in which
$\mathbf R_j$ is rank deficient.
In addition, for fixed $(\underline{\mathbf s},\underline{\mathbf V},\underline{\boldsymbol\Gamma})$, the $\mathbf Z_j$ optimization subproblem remains strictly concave because $\mathbf F_j\succ\mathbf0$ and $\delta(\boldsymbol\Gamma_j)\succeq\mathbf0$.


\subsubsection{Update of the auxiliary variables}
For fixed $(\underline{\mathbf s},\underline{\mathbf V})$, the optimal auxiliary variable induced by the RIT is updated by the current SINR matrix
\begin{equation}
\boldsymbol\Gamma_j^\star = \mathbf R_j =
\mathbf S_j^{\mathrm H}\mathbf F_j^{-1}\mathbf S_j.
\label{eq:gamma_update}
\end{equation}

For fixed $(\underline{\mathbf s},\underline{\mathbf V},\underline{\boldsymbol\Gamma})$, taking the first-order optimality condition of~\eqref{eq:scaled_sefp_objective} with respect to $\mathbf Z_j$ gives
\begin{equation}
\mathbf F_j\mathbf Z_j +
\mathbf S_j\mathbf S_j^{\mathrm H}
\mathbf Z_j\delta(\boldsymbol\Gamma_j)
=  \sqrt{\omega_{j,s_j}}\,
\mathbf S_j\rho(\boldsymbol\Gamma_j).
\label{eq:U_sylvester}
\end{equation}
\begin{remark}
Equation~\eqref{eq:U_sylvester} is a generalized Sylvester equation. When the RIT induced auxiliary  variable is updated according to~\eqref{eq:gamma_update}, it admits the following closed-form solution
\begin{equation}
\mathbf Z_j^\star =
\sqrt{\omega_{j,s_j}}\,
\mathbf F_j^{-1}\mathbf S_j
(\mathbf I+\mathbf R_j)^{-1/2}.
\label{eq:U_closed_form}
\end{equation}
\end{remark}
\begin{proof}
Substituting~\eqref{eq:U_closed_form}, 
and using $\mathbf S_j^{H}\mathbf F_j^{-1}\mathbf S_j=\mathbf R_j$, the left-hand-side of~\eqref{eq:U_sylvester} becomes
$$ \sqrt{\omega_{j,s_j}} \mathbf S_j \left[
 \mathbf (\mathbf I+\mathbf R_j)^{-1/2}
 +\mathbf R_j\mathbf (\mathbf I+\mathbf R_j)^{-1/2}\delta(\mathbf R_j)\right]. $$ 
Leveraging the scalar identity ${1+t\delta(t)}/{\sqrt{1+t}}=\rho(t),$ and the matrix function property in Section~\ref{subsec:Hermitian functional calculus} directly proves~\eqref{eq:U_closed_form}. 
\end{proof}

\subsubsection{Edge-wise decomposition of the SEFP surrogate} 
We next optimize the scheduling variables and beamformers while holding the auxiliary variables $(\underline{\boldsymbol\Gamma},\underline{\mathbf Z})$ fixed. Define the transmitter activation variable
$a_i\in\{0,1\}\triangleq\sum_{j\in\mathcal L_i}q_{ji}.$
The interference covariance at receiver $k$ can then be written as
\begin{equation}
\mathbf F_k = \sigma_k^2\mathbf I + \sum_{i\in\mathcal I}
(a_i-q_{ki})
\mathbf H_{ki}\mathbf V_i\mathbf V_i^{ H}\mathbf H_{ki}^{H}.
\label{eq:F_binary_form}
\end{equation}
Substituting~\eqref{eq:F_binary_form} into~\eqref{eq:scaled_sefp_objective} and collecting all terms associated with a candidate edge $(j,i)$ gives
\begin{equation}
\bar f_{\mathrm{SE}} = C(\underline{\mathbf Z}) + \sum_{j\in\mathcal J}\sum_{i\in\mathcal K_j}
q_{ji}\psi_{ji}(\mathbf V_i),
\label{eq:edge_decomposition}
\end{equation}
where $C(\underline{\mathbf Z})=-
\sum_{k\in\mathcal J}\sigma_k^2
\tr\!\left(\mathbf Z_k^{H}\mathbf Z_k\right)$
denotes the constant term with respect to $\mathbf V_i$,
and
\begin{equation}
\psi_{ji}(\mathbf V_i)
=
2\Re\tr\!\left(\mathbf V_i^{ H}\mathbf G_{ji}\right)
-
\tr\!\left(\mathbf V_i^{H}\mathbf A_{ji}\mathbf V_i\right),
\label{eq:edge_surrogate}
\end{equation}
where the linear coefficient matrix is
\begin{equation}
\mathbf G_{ji}
=
\sqrt{\omega_{ji}}\,
\mathbf H_{ji}^{ H}\mathbf Z_j\rho(\boldsymbol\Gamma_j),
\label{eq:G_ji}
\end{equation}
and the semi-definite  quadratic coefficient matrix is
\begin{equation}
\mathbf A_{ji}
={}
\mathbf H_{ji}^{ H}
\mathbf Z_j\delta(\boldsymbol\Gamma_j)
\mathbf Z_j^{H}\mathbf H_{ji}
+
\sum_{\substack{k\in\mathcal J, k\neq j}}
\mathbf H_{ki}^{ H}\mathbf Z_k\mathbf Z_k^{H}\mathbf H_{ki}.
\label{eq:A_ji}
\end{equation}
\begin{remark}
By comparing with the FPLinQ surrogate~\cite[Eq. 32(b)]{FPPart3},  
a distinguishing property of the SEFPLinQ surrogate is that $\mathbf A_{ji}$ generally depends on both the transmitter and the candidate receiver, while the corresponding quadratic coefficient in FPLinQ depends only on the transmitter index. 
This occurs because the signal sent by transmitter $i$ is weighted by $\delta(\boldsymbol\Gamma_j)$ at its intended receiver $j$, yet by the identity matrix when it appears as interference at the remaining receivers.
\end{remark}

\subsubsection{Tentative beamformer for each candidate edge and the first step matching}
Observing~\eqref{eq:edge_decomposition} that the optimization of the beamformer of each link is mathematically independent of any other links,
for each feasible association edge $(j,i)$, define the tentative beamformer $\widetilde{\mathbf V}_{ji}$ as the solution to
\begin{equation}\label{prob:tentative_beamformer}
\underset{\mathbf V_i}{\text{arg  max}}\quad
2\Re\tr\!\left(\mathbf V_i^{ H}\mathbf G_{ji}\right)
-\tr\!\left(\mathbf V_i^{H}\mathbf A_{ji}\mathbf V_i\right),
\end{equation}
under the transmit power constraint~\eqref{prob:wsr_power}.
By introducing a nonnegative Lagrange multiplier $\mu_{ji}$, and leveraging the first-order KKT optimality condition, $\widetilde{\mathbf V}_{ji}$ can be computed by
\begin{equation}
\widetilde{\mathbf V}_{ji}
=
(\mathbf A_{ji}+\mu^{\star}_{ji}\mathbf I_{N_t})^{-1}\mathbf G_{ji},
\label{eq:tentative_beamformer}
\end{equation}
where in practice, the optimal Lagrange multiplier $\mu^{\star}_{ji}$ can be efficiently obtained by one-dimensional bisection search.

Substituting $\widetilde{\mathbf V}_{ji}$ into~\eqref{eq:edge_surrogate}, we define
\begin{equation}
\lambda_{ji}^{\mathrm{SE}}
\triangleq
2\Re\tr\!\left(
\widetilde{\mathbf V}_{ji}^{ H}\mathbf G_{ji}
\right) - \tr\!\left(
\widetilde{\mathbf V}_{ji}^{H}
\mathbf A_{ji}\widetilde{\mathbf V}_{ji}
\right)=\psi_{ji}( \widetilde{\mathbf V}_{ji})
\label{eq:sefp_edge_weight}
\end{equation}
to serve as the first step matching's edge weight of the corresponding link associated with the transmitter $i$ and the receiver $j$. 
Next, maximizing the SEFPLinQ surrogate $\bar f_{\mathrm{SE}}$ in~\eqref{eq:edge_decomposition} of the entire network is equivalent with
\begin{equation}\label{prob:first_matching}
\underset{\mathbf q}\max\quad
\sum_{j\in\mathcal J}\sum_{i\in\mathcal K_j}
q_{ji}\lambda_{ji}^{\mathrm{SE}} \qquad\mathrm{s.t.}  \ \ \text{(\ref{eq:matching_constraints})}.
\end{equation}
Problem~\eqref{prob:first_matching} is essentially a maximum weighted bipartite matching problem, which is well-studied and could be solved optimally in polynomial time by centralized and distributed algorithms such as~\cite{Hungarian} and~\cite{Auction}.

Let $\widetilde{\underline{\mathbf{q}}} \triangleq {\{\widetilde{q}_{ji}|\ i\in\mathcal I,j\in\mathcal J\}}$ denote the solution to~\eqref{prob:first_matching}. The beamformer of transmitter $i$ can then be updated as
\begin{equation}
\mathbf V^{\star}_i
=
\sum_{j\in\mathcal L_i}
\widetilde q_{ji}\widetilde{\mathbf V}_{ji}.
\label{eq:first_matching_beamformer}
\end{equation}
Because at most one $\widetilde q_{ji}=1$,~\eqref{eq:first_matching_beamformer} selects at most one candidate beamformer for every transmitter. If no edge incident to transmitter $i$ is selected, then $\mathbf V^{\star}_i=\mathbf0$.
\begin{remark}[Why the SEFP coordinate is advantageous for scheduling]
\label{rem:sefp_xmmse_scheduling_structure}
The edge decomposition in~\eqref{eq:edge_decomposition} makes the conceptual distinction discussed in Section~\ref{subsec:xmmse_special_case} explicit.
Once $(\underline{\boldsymbol\Gamma},\underline{\mathbf Z})$ are fixed, the same pair $(\boldsymbol\Gamma_j,\mathbf Z_j)$ is reused for all candidate transmitters of receiver $j$, and the scheduling variables enter the SEFP surrogate through the form $q_{ji}\psi_{ji}(\mathbf V_i)$. 
Consequently, each tentative beamformer $\widetilde{\mathbf V}_{ji}$ and its scalar score $\lambda_{ji}^{\mathrm{SE}}$ can be computed independently before matching. The only remaining coupling is imposed by the one-to-one matching constraints in~\eqref{prob:first_matching}.

The direct XMMSE coordinate in~\eqref{eq:xmmse_functional_objective}, however, does not expose this structure. 
The obstruction arises from its interference factor $\omega_{j,s_j}\mathbf F_j$. 
Specifically, $\omega_{j,s_j}=\sum_{i\in\mathcal K_j}q_{ji}\omega_{ji}$, whereas $\mathbf F_j$ depends on the other selected edges. Their product therefore expands as
\begin{equation}
\autoeq{
\sigma_j^2\sum_{i\in\mathcal K_j}q_{ji}\omega_{ji}\mathbf I_{N_r}+\sum_{i\in\mathcal K_j}
\sum_{\substack{k\in\mathcal J\\k\ne j}}
\sum_{m\in\mathcal K_k}
q_{ji}q_{km}\omega_{ji}
\mathbf H_{jm}\mathbf V_m\mathbf V_m^H\mathbf H_{jm}^H.}
\label{eq:xmmse_scheduling_coupling_sefplinq} 
\end{equation}
When~\eqref{eq:xmmse_scheduling_coupling_sefplinq} is inserted into the quadratic integral in~\eqref{eq:xmmse_functional_objective}, the products $q_{ji}q_{km}$ make the contribution of edge $(j,i)$ dependent on which interfering edges $(k,m)$ are simultaneously selected. Hence, no schedule-independent score analogous to $\lambda_{ji}^{\mathrm{SE}}$ is available before matching.
 
\end{remark}

\begin{remark}[Premature turn-off alleviated]The premature turn-off phenomenon is such that the performance of a direct block coordinate type method may permanently deactivate a transmitter.
In particular, if $\mathbf V_i=\mathbf0$, all actual-rate matching weights associated with transmitter $i$ are zero when the beamformers are held fixed, preventing the transmitter from being selected again.

A remarkable ability possessed by the FPLinQ is that it alleviates the premature turn-off efficiently, which can be inherited by the SEFPLinQ. To be specific, the tentative beamformer in~\eqref{eq:tentative_beamformer}
is recomputed for every feasible edge and does not explicitly depend on the current value of $\underline{\mathbf V}$ and $\underline{\mathbf s}$. 
Therefore, even if $\mathbf V_i^{(t)}=\mathbf0$, a nonzero tentative beamformer can be generated whenever the current linear coefficient matrix 
$\mathbf G_{ji}^{(t)}\neq\mathbf 0$ in~\eqref{eq:G_ji}
for at least one candidate receiver $j$.
An inactive transmitter can consequently be reactivated by the matching in~\eqref{prob:first_matching}.

\end{remark}

\subsubsection{ The second step matching}
After updating $\underline{\mathbf V}$,
the second step matching plays a role in updating $\underline{\mathbf s}$ by directly optimizing the original rate objective.

Define the set of transmitters activated by the first matching as
$ \overline{\mathcal I} \triangleq
\left\{i\in\mathcal I:\normF{\mathbf V^{\star}_i}>0\right\}.$
For every feasible edge $(j,i)$ with $i\in\overline{\mathcal I}$, the corresponding interference-plus-noise covariance matrix is $\overline{\mathbf F}_{ji} \triangleq
\sigma_j^2\mathbf I + \sum_{\ell\in\overline{\mathcal I}\setminus\{i\}}
\mathbf H_{j\ell}\mathbf V^{\star}_{\ell}
{\mathbf V^{\star}_{\ell}}^{ H}\mathbf H_{j\ell}^{ H},$
and the corresponding tentative actual rate is
$$ r_{ji}(\underline{\mathbf V}^{\star}) = \log\det\!\left(
\mathbf I_d+ {\mathbf V^{\star H}_i}\mathbf H_{ji}^{ H}
\overline{\mathbf F}_{ji}^{-1}
\mathbf H_{ji}\mathbf V^{\star}_i
\right).$$
The update of $\underline{\mathbf{s}}$ turns to adjusting the transmitter--receiver service relation to maximize the WSR, that is
\begin{equation}\label{prob:second_matching}
\underset{\mathbf q}{\max}\quad
\sum_{j\in\mathcal J}\sum_{i \in \mathcal{K}_j \cap \bar{\mathcal{I}}}
q_{ji}\left(\omega_{ji}r_{ji}(\underline{\mathbf V}^{\star})\right)
 \qquad {\rm s.t.}  \ \ \text{(\ref{eq:matching_constraints})},
\end{equation}
which is still a weighted bipartite matching problem.
Let $\underline{\mathbf{q}}^{+} \triangleq {\{{q}^{+}_{ji}|\ i\in\overline{\mathcal I},j\in\mathcal J\}}$ denote the solution.
The corresponding updated scheduling variables are
\begin{equation}
s_j^{\star}=\begin{cases}
i, & q_{ji}^{+}=1\text{ for some } i \in \mathcal{K}_j \cap \bar{\mathcal{I}},\\
\varnothing, & \sum_{i \in \mathcal{K}_j \cap \bar{\mathcal{I}}}q_{ji}^{+}=0.
\end{cases}
\label{eq:schedule_update}
\end{equation}

Algorithm~\ref{alg:sefplinq} summarizes the overall SEFPLinQ approach.
\begin{algorithm}[t]
\caption{Proposed SEFPLinQ Strategy for Flexible Association D2D Link Scheduling and Beamforming}
\label{alg:sefplinq}
\begin{algorithmic}[1]
\State Initialize all the variables to feasible values;
\Repeat
    \State Update $\underline{\boldsymbol\Gamma}$ according to~\eqref{eq:gamma_update};
    \State Update $\underline{\mathbf Z}$ according to~\eqref{eq:U_closed_form};
    \State Solve the first weighted bipartite matching~\eqref{prob:first_matching};
    \State Update $\underline{\mathbf V}$ according to~\eqref{eq:first_matching_beamformer};
    \State Solve the second weighted bipartite matching~\eqref{prob:second_matching};
    \State Update $\underline{\mathbf s}$ according to~\eqref{eq:schedule_update};
\Until{\textit{the weighted sum rate converges};}
\end{algorithmic}
\end{algorithm}

\section{SIMULATION RESULTS}\label{sec:numerical results}
\subsection{Common Simulation Setup}

\begin{table}[t]
    \centering
    \caption{Common simulation parameters for the network settings.}
    \label{tab:sim_parameters}
    \renewcommand{\arraystretch}{1.05}
    \setlength{\tabcolsep}{3.0pt}
    \footnotesize
    \begin{tabular*}{\columnwidth}{@{\extracolsep{\fill}}lclc@{}}
        \toprule
        \textbf{Parameter} & \textbf{Value} &
        \textbf{Parameter} & \textbf{Value} \\
        \midrule
        Area
        & $1~\mathrm{km}\!\times\!1~\mathrm{km}$
        & Channel model
        & ITU-1411 \\

        Bandwidth
        & $5~\mathrm{MHz}$
        & Carrier frequency
        & $2.4~\mathrm{GHz}$ \\

        Antenna gain
        & $2.5~\mathrm{dBi}$
        & Device height
        & $1.5~\mathrm{m}$ \\

        Noise PSD
        & $-169~\mathrm{dBm/Hz}$
        & Max. Tx power
        & $20~\mathrm{dBm}$ \\

        Shadowing std.
        & $10~\mathrm{dB}$
        & Noise figure
        & $7~\mathrm{dB}$ \\

        \bottomrule
    \end{tabular*}
\end{table}
Unless otherwise specified, all simulations use the common D2D network setting parameters summarized in Table~\ref{tab:sim_parameters}, in line with those in~\cite{FPPart3}.
Specifically, two network association configurations are considered: the fixed single-association and flexible-association settings.  
For the former, each transmitter is paired with one dedicated receiver, and the direct link distance is independently and uniformly generated between $[2,65]$ m.
For the latter, we deploy $|\mathcal J|=100$ receivers and $|\mathcal I|=300$ transmitters.
Three local transmitters are first generated for every receiver, each with a transmitter--receiver distance uniformly distributed over $[2,65]$ m, yielding $300$ primary candidate edges.
We then randomly select $100$ transmitters and assign each selected transmitter one additional candidate edge to its nearest receiver that is not already associated with it.
Consequently, every flexible-association realization contains exactly $400$ feasible transmitter--receiver candidate edges. 
Each transmitter and each receiver participate in at most one active link in each scheduling interval.

Two network utility metrics are considered.
For the WSR metric, all feasible links are assigned fixed unit weights.
For the proportional fairness (PF) metric, link weights are dynamically updated according to the PF criterion, which, in the long term, is equivalent to maximizing the sum of the logarithms of the average link rates.
$300$ scheduling slots are simulated for each realization. 
Each WSR subproblem is solved from a cold start without warm start across PF slots.

For reproducibility and fairness, all algorithms compared in each experiment are evaluated over the same network and channel realizations. 
Each experiment performs 30 independent channel realizations. 
In the fixed-association SISO setting, FPLinQ and SEFPLinQ are initialized with all links transmitting at $P_{\max}$, whereas
the initialization follows the same greedy weight-ordered one-to-one scheduling rule with normalized full power beamformers.

\subsection{Fixed Single-association D2D Networks}
{To validate the interconnections of the two FP-type and the two MMSE-type algorithms delineated in Proposition~\ref{coro:Equality_of_surro} and Remark~\ref{euqi_bound}, Fig.~\ref{fig:WSRvsIternum4alg} compares the WSR trajectories of the four algorithms under the WSR metric. 
For both the SISO and the multi-data-stream MIMO cases, the SEFP and XMMSE WSR trajectories coincide within numerical precision throughout the iterations, as do FP and WMMSE, corroborating the equivalence.
Notably, under both antenna configurations, SEFP and XMMSE exhibit a visible gain compared to closed-form FP and WMMSE.}
\begin{figure}[htbp] 
  \centering 
  \hspace*{-0.4cm}\includegraphics[scale=0.59]{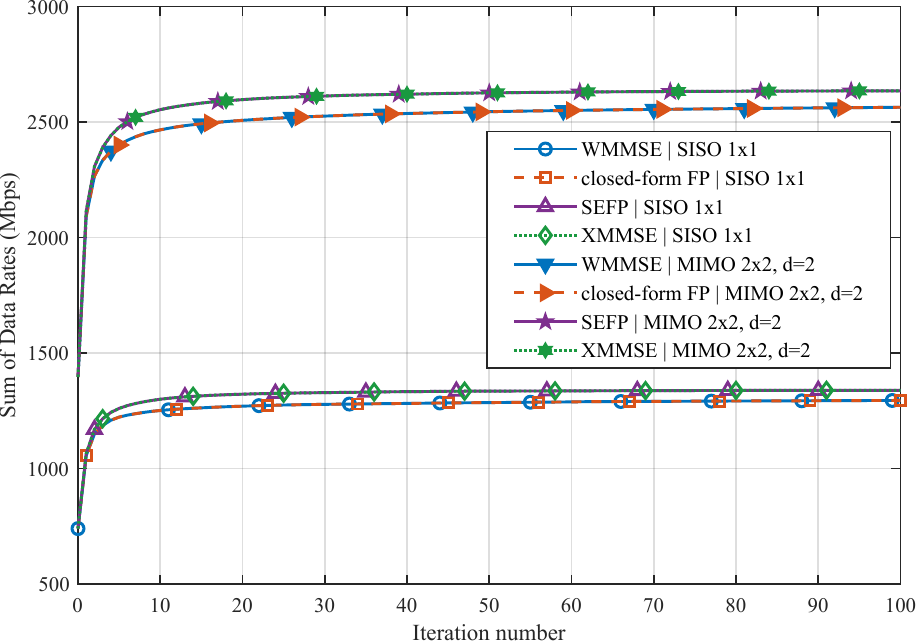} 
  \caption{{WSR versus iteration number for the two FP-type and the two MMSE-type algorithms in a fixed single-association D2D network.}} 
  \label{fig:WSRvsIternum4alg} %
\end{figure}

Fig.~\ref{fig:FixAss_WSRvsLinkNum} compares the WSR performance in the fixed single-association SISO D2D network.
For each realization, we iterate FPLinQ and SEFPLinQ until convergence, up to the maximum of $200$ iterations.
In particular, SEFPLinQ consistently achieves the highest sum rate over all considered network sizes, slightly but steadily outperforming FPLinQ.
Both FP-type algorithms substantially outperform the conventional heuristic scheduling baselines.
\begin{figure}[htbp] 
  \centering 
  \hspace*{-0.4cm}\includegraphics[scale=0.67]{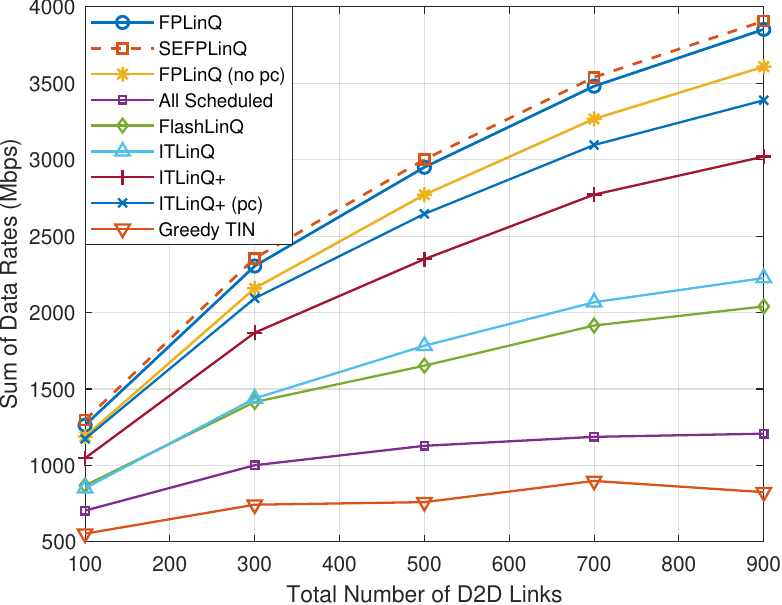} 
  \caption{WSR versus the total number of links for the fixed single-association D2D network.} 
  \label{fig:FixAss_WSRvsLinkNum} 
\end{figure}

Fig.~\ref{fig:FixAss_PF_cdf} considers the fixed single-association network with $100$ D2D links under PF weight updating, and compares the cumulative distribution of the resulting long-term link rates. Compared with FPLinQ, SEFPLinQ shifts the CDF toward higher rates mainly in the low- and medium-rate regimes while maintaining a competitive high-rate performance. 
Combining with the numerical result in Table~\ref{tab:sum_log_utility}, SEFPLinQ shows a better balance among the links and a higher overall log-utility, illustrating the benefit of the SEFP surrogate when fairness is incorporated into scheduling and power control.
\begin{figure}[htbp] 
  \centering 
  \hspace*{-0.4cm}\includegraphics[scale=0.67]{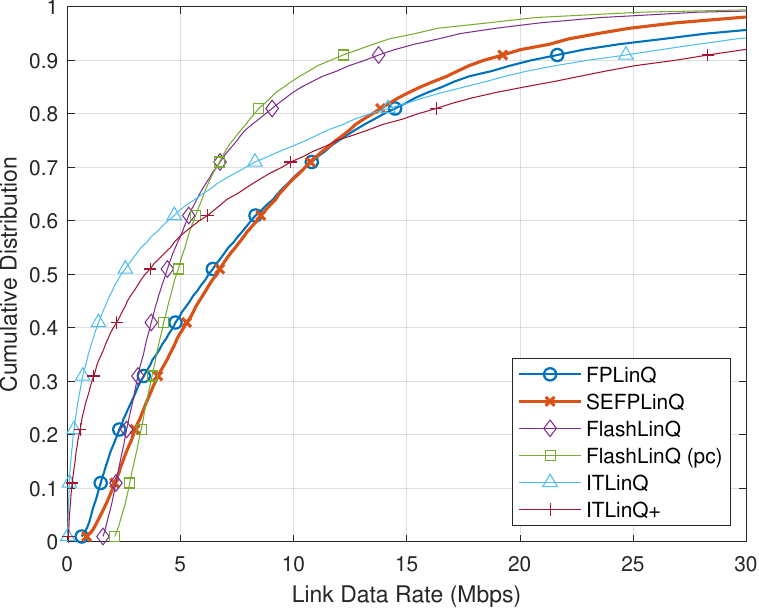} 
  \caption{Log-utility CDF for the fixed single-association D2D network.} 
  \label{fig:FixAss_PF_cdf} 
\end{figure}

\subsection{Flexible-association D2D Networks}
We then turn to the ﬂexible-association case, where for SEFPLinQ and FPLinQ the second matching is indispensable. 
With also the consideration of the MIMO case, FlashLinQ, ITLinQ, and ITLinQ+ are excluded because of their inapplicability.
Fig.~\ref{fig:FleAss_overcomepremature} evaluates PF scheduling in the flexible-association network for both the SISO and $2\times2$ MIMO settings. 
BCD suffers a substantial performance loss because of the premature turning-off, whereas both FPLinQ and SEFPLinQ avoid it effectively. 
Moreover, SEFPLinQ achieves a rate cumulative distribution that is consistently competitive with, and generally better than that of FPLinQ, demonstrating that SEFPLinQ inherits the robust scheduling capability of FPLinQ while further improving the resulting network utility.

\begin{figure}[htbp] 
  \centering 
  \hspace*{-0.4cm}\includegraphics[scale=0.69]{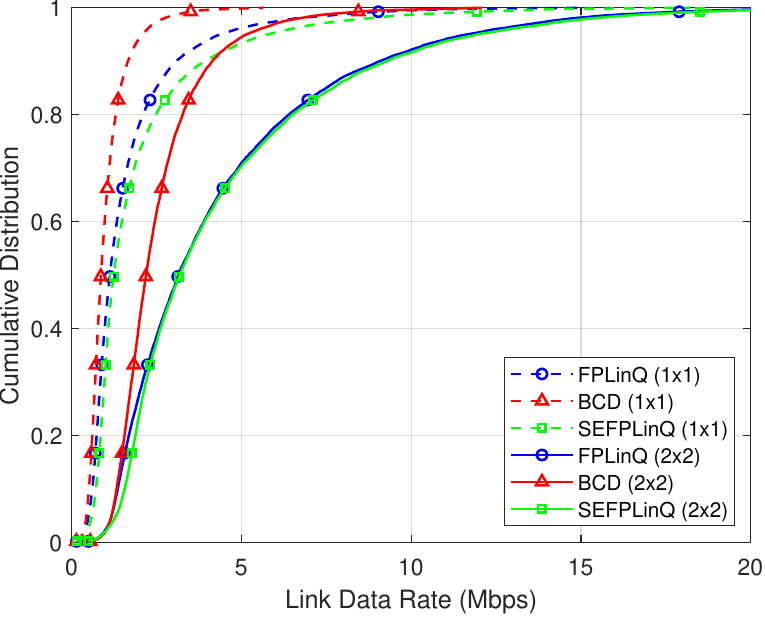} 
  \caption{Log-utility maximization for the ﬂexible-association D2D network:
SEFPLinQ vs. FPLinQ and BCD.} 
  \label{fig:FleAss_overcomepremature} 
\end{figure}

Fig.~\ref{fig:FleAss_matrixbenifit} compares matrix FPLinQ and SEFPLinQ with the single-stream vector FP method for $2\times2$, $4\times4$, and $8\times8$ MIMO configurations.
FPLinQ and SEFPLinQ employ spatial streams, whereas vector FP is restricted to one stream. 
The difference is small in the $2\times2$ case but becomes increasingly pronounced as the antenna dimension grows, confirming the benefit of multi-stream matrix optimization.
Moreover, SEFPLinQ achieves a higher overall log-utility than FPLinQ across the considered MIMO configurations, where the advantage becomes more evident in the larger-dimensional cases.

\begin{figure}[htbp] 
  \centering 
  \hspace*{-0.4cm}\includegraphics[scale=0.67]{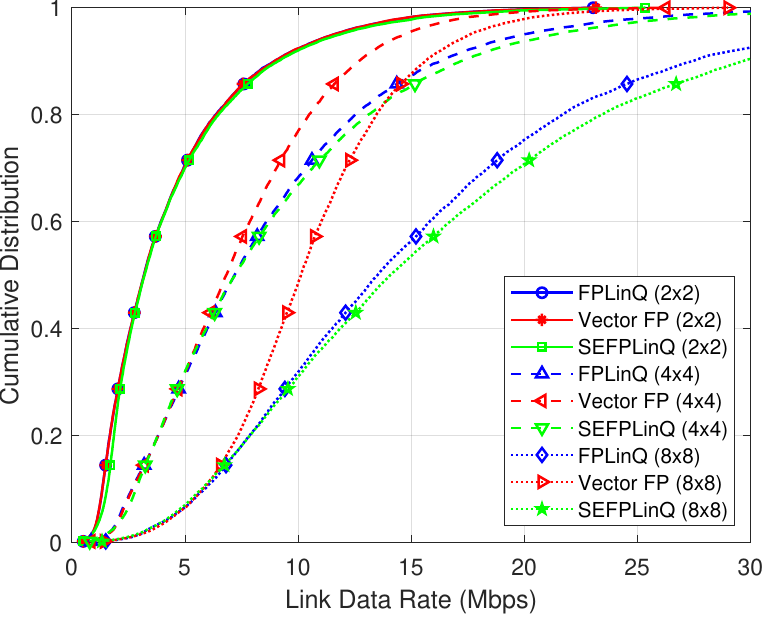} 
  \caption{Log-utility maximization for the ﬂexible-association D2D network:
SEFPLinQ vs. FPLinQ and Vector FP.} 
  \label{fig:FleAss_matrixbenifit} 
\end{figure}

Finally, Fig.~\ref{fig:convergence} compares the convergence behavior of FPLinQ and SEFPLinQ for equal-weight sum-rate maximization with no PF outer loop employed in the flexible-association network.
For each antenna configuration and each channel realization, both FPLinQ and SEFPLinQ execute $100$ iterations without early stopping.
SEFPLinQ exhibits both faster rate improvement and a larger sum-rate gain throughout the iteration horizon.
In particular, for the $8\times8$ case, when taking the $100$-iteration value as reference, SEFPLinQ reaches $95\%$ of this gain after only $26$ iterations, in contrast to $62$ iterations for FPLinQ; at iteration $100$, the corresponding average sum-rate increments are $13,068.2$ Mbps and $11,128.7$ Mbps, respectively, representing a $17.43\%$ improvement.
\begin{figure}[htbp] 
  \centering 
  \hspace*{-0.4cm}\includegraphics[scale=0.66]{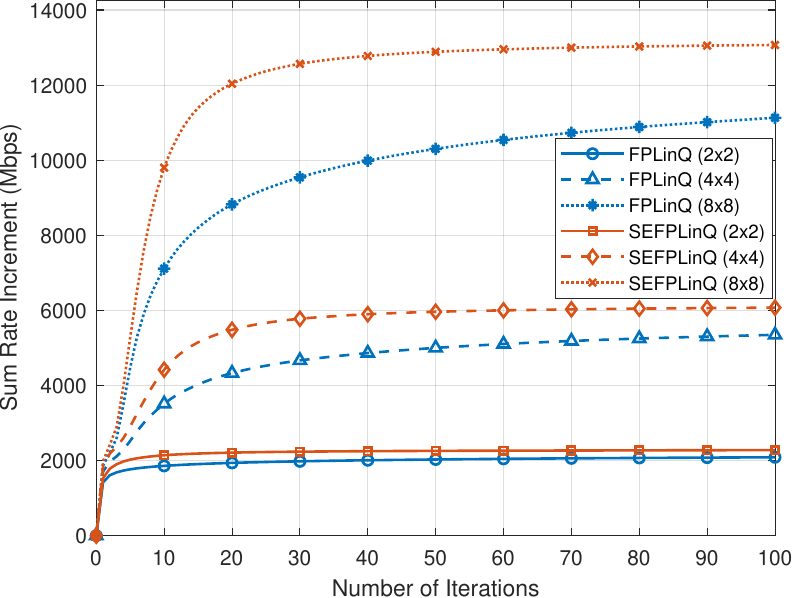} 
  \caption{Convergence of SEFPLinQ and FPLinQ in maximizing the sum rate for the flexible-association D2D network.} 
  \label{fig:convergence} 
\end{figure}

\begin{table}[t]
\centering
\caption{Network log-utility comparison under fixed and flexible association.}
\label{tab:log_utility}
\footnotesize
\renewcommand{\arraystretch}{1.03}

\begin{tabular*}{\columnwidth}{@{\extracolsep{\fill}}lrlr@{}}
\toprule
\multicolumn{4}{c}{\textbf{Fixed single-association}} \\
\midrule
\textbf{Algorithm} & \textbf{$\mathbf{1\!\times\!1}$}
& \textbf{Algorithm} & \textbf{$\mathbf{1\!\times\!1}$} \\
\midrule
\textbf{SEFPLinQ} & \textbf{185.4}
& FPLinQ & 177.1 \\

FPLinQ (no pc) & 146.1
& FlashLinQ & 158.4 \\

FlashLinQ (pc) & 167.0
& ITLinQ+ & 104.4 \\
\bottomrule
\end{tabular*}

\vspace{4pt}

\begin{tabular*}{\columnwidth}{@{\extracolsep{\fill}}lrrlrr@{}}
\toprule
\multicolumn{6}{c}{\textbf{Flexible-association}} \\
\midrule
\textbf{Algorithm}
& \multicolumn{1}{c}{$\mathbf{1\!\times\!1}$}
& \multicolumn{1}{c}{$\mathbf{2\!\times\!2}$}
& \textbf{Algorithm}
& \multicolumn{1}{c}{$\mathbf{4\!\times\!4}$}
& \multicolumn{1}{c}{$\mathbf{8\!\times\!8}$} \\
\midrule

\textbf{SEFPLinQ}
& \textbf{146.8}
& \textbf{496.6}
& \textbf{SEFPLinQ}
& \textbf{783.0}
& \textbf{1042.1} \\

FPLinQ
& 87.0
& 479.6
& FPLinQ
& 774.2
& 1027.7 \\

Vector FP
& 87.0
& 479.0
& Vector FP
& 737.5
& 910.2 \\

BCD
& $-48.9$
& 328.6
&
&
& \\

\bottomrule
\end{tabular*} \label{tab:sum_log_utility}
\end{table}

\section{CONCLUSION}\label{sec:conclu}
This paper developed a matrix SEFP framework for the weighted sum-log-determinant maximization problem with matrix ratios, going beyond the classical matrix FP paradigm.
Such a matrix SEFP is a nontrivial generalization of the previously proposed scalar SEFP with RIT, leveraging the mathematical tools of Hermitian functional calculus.
We further established an insightful interconnection between SEFP and XMMSE, analogous to the FP-WMMSE relationship, showing that they induce identical MM surrogates under exact auxiliary variable updates, such that XMMSE turns out to be a particular realization within the broad SEFP family. 
Consequently, building upon the matrix SEFP, we developed the SEFPLinQ algorithm for flexible-association MIMO D2D networks, with numerical results demonstrating the state-of-the-art weighted sum-rate performance.

\appendices
\section{Proof of Theorem~\ref{thm:matrix_rit}} \label{pro:mRIT_theorem}
To provide insight on how the matrix RIT is obtained, we conduct a constructive proof here. 

It suffices to establish the transform for one matrix ratio, since the auxiliary matrices are separable across $m$ for fixed $\mathbf{x}$.
We therefore suppress the index $m$ and write $
\mathbf{R}=\mathbf{S}^{H}\mathbf{F}^{-1}\mathbf{S}\succeq\mathbf{0},$
with $\mathbf{S}\in\mathbb{C}^{n\times d}$ and $\mathbf{F}\in\mathbb{H}_{++}^{n}$.
Following the eigen-decomposition of $\boldsymbol{\Gamma}$ in~\eqref{eq:gamma_evd},
the scalar lower bound property in Corollary~\ref{cor:scalar_rit_lower_bound} suggests the directional lifting as
\begin{equation}\label{eq:directional_lower_bound}
\Phi_{\boldsymbol{\Gamma}}(\mathbf{R})
\triangleq
\sum_{i=1}^{d}
\boldsymbol{\xi}_i^{H}\phi_{\gamma_i}(\mathbf{R})\boldsymbol{\xi}_i,
\end{equation}
where $\phi_{\gamma_i}(\mathbf{R}) \triangleq  c(\gamma_i) \mathbf{R}\left(\gamma_i^2\mathbf{I}+b(\gamma_i)\mathbf{R}\right)^{-1}$ is obtained by applying the scalar function $\phi_{\gamma_i}(\cdot)$ eigenvalue-wise to the Hermitian matrix $\mathbf{R}$ according to the matrix function property in Section~\ref{subsec:Hermitian functional calculus}.

Naturally, analogous to Corollary~\ref{cor:scalar_rit_lower_bound}, we have the following lemma.
\begin{lemma}[Matrix RIT's Lower-Bound Property] \label{lem:mRIT_LB}
For every fixed $\boldsymbol{\Gamma}\succ\mathbf{0}$, $\Phi_{\boldsymbol{\Gamma}}(\mathbf{R})$ is a globally valid lower bound of the log-determinant, that is
\begin{equation}
\Phi_{\boldsymbol{\Gamma}}(\mathbf{R})
\leq
\log\det(\mathbf{I}+\mathbf{R}), \qquad \forall \ \mathbf{R} \succeq0.
\label{eq:matrix_rit_lower_bound}
\end{equation}
For $\mathbf{R}\succ\mathbf{0}$, equality in \eqref{eq:matrix_rit_lower_bound} holds if and only if $\boldsymbol{\Gamma}=\mathbf{R}$.
If $\mathbf{R}$ is singular, every $\boldsymbol{\Gamma}$ of the form~\eqref{eq:gamma_opt_singular} attains equality.
\end{lemma}
\begin{proof}
See Appendix~\ref{proof_of_mRIT_LB_lem}.
\end{proof}

Lemma~\ref{lem:mRIT_LB} captures the part of the construction of the matrix RIT that is inherited directly from the scalar RIT: each eigenvalue $\gamma_i$ specifies a scalar RIT contact value, whereas the corresponding eigenvector $\boldsymbol{\xi}_i$ specifies the direction along which that scalar bound is evaluated.
The remaining goal is to algebraically assemble the direction-wise parameters $\{(\gamma_i,\boldsymbol{\xi}_i)\}_{i=1}^{d}$ back into one entire matrix, which should be the matrix $\boldsymbol{\Gamma}$ exactly, without explicitly retaining the individual spectral directions in the final transform.

To this end, recall the operator $\mathcal{D}_{\boldsymbol{\Gamma}}$ in~\eqref{eq:Dm_operator} with the term index suppressed. Since $\mathcal{D}_{\boldsymbol{\Gamma}}$ is invertible (proved in Appendix~\ref{pro:coro_oper_repr}), define the unique matrix $\mathbf{W}_{\boldsymbol{\Gamma}}$ by
\begin{equation}
\mathbf{W}_{\boldsymbol{\Gamma}} \triangleq
\mathcal{D}_{\boldsymbol{\Gamma}}^{-1}
\!\left(\mathbf{S}c(\boldsymbol{\Gamma})^{1/2}\right),
\label{eq:WGamma_operator_inverse}
\end{equation}
or equivalently,
\begin{equation}
\mathcal{D}_{\boldsymbol{\Gamma}}(\mathbf{W}_{\boldsymbol{\Gamma}})=\mathbf{F}\mathbf{W}_{\boldsymbol{\Gamma}}\boldsymbol{\Gamma}^{2}
+\mathbf{S}\mathbf{S}^{H}\mathbf{W}_{\boldsymbol{\Gamma}}b(\boldsymbol{\Gamma})
=
\mathbf{S}c(\boldsymbol{\Gamma})^{1/2}.
\label{eq:operator_equation_Z}
\end{equation}
Right-multiplying~\eqref{eq:operator_equation_Z} by the $i$-th eigenvector $\boldsymbol{\xi}_i$ of $\boldsymbol{\Gamma}$ and using the spectral definitions of $\boldsymbol{\Gamma}^{2}$, $b(\boldsymbol{\Gamma})$, and $c(\boldsymbol{\Gamma})^{1/2}$ gives
\begin{equation}
\left(\gamma_i^2\mathbf{F}
+b(\gamma_i)\mathbf{S}\mathbf{S}^{H}\right)
\mathbf{\mathbf{W}_{\boldsymbol{\Gamma}}}\boldsymbol{\xi}_i
=
\sqrt{c(\gamma_i)}\,\mathbf{S}\boldsymbol{\xi}_i,
\label{eq:operator_direction_i}
\end{equation}
which directly shows the single operator equation~\eqref{eq:operator_equation_Z} simultaneously contains the $d$ direction-wise linear systems associated with the scalar RIT parameters $\{\gamma_i\}_{i=1}^{d}$.

Next, using the push-through identity $(\lambda\mathbf{I} + \mathbf{A}\mathbf{B})^{-1} \mathbf{A} = \mathbf{A}(\lambda\mathbf{I} + \mathbf{B}\mathbf{A})^{-1}$~\cite[Chapter 11]{IntroductionLA}, we obtain
\begin{equation}
\mathbf{S}^{H}
\left(\gamma_i^2\mathbf{F}+b(\gamma_i)\mathbf{S}\mathbf{S}^{H}\right)^{-1}
\mathbf{S}
=
\mathbf{R}
\left(\gamma_i^2\mathbf{I}+b(\gamma_i)\mathbf{R}\right)^{-1}.
\label{eq:push_through}
\end{equation}

Combining~\eqref{eq:push_through} with the definition of $\phi_{\gamma_i}(\mathbf{R})$ and~\eqref{eq:operator_direction_i} yields
\begin{align}
\boldsymbol{\xi}_i^{H}\phi_{\gamma_i}(\mathbf{R})\boldsymbol{\xi}_i
&=
c(\gamma_i)(\mathbf{S}\boldsymbol{\xi}_i)^{H}
\left(\gamma_i^2\mathbf{F}+b(\gamma_i)\mathbf{S}\mathbf{S}^{H}\right)^{-1}
(\mathbf{S}\boldsymbol{\xi}_i)\nonumber\\
&=
\sqrt{c(\gamma_i)}\,
(\mathbf{S}\boldsymbol{\xi}_i)^{H}\mathbf{W}_{\boldsymbol{\Gamma}}\boldsymbol{\xi}_i.
\label{eq:directional_energy}
\end{align}
Equation~\eqref{eq:directional_energy} reveals that each direction-wise scalar RIT term can be recovered from the common matrix $\mathbf{W}_{\boldsymbol{\Gamma}}$ through its operation along the corresponding eigen-direction $\boldsymbol{\xi}_i$, which motivates the following representation straightforwardly.

Summing~\eqref{eq:directional_energy} over all directions 
yields
\begin{equation}
\autoeq{%
\begin{aligned}
\Phi_{\boldsymbol{\Gamma}}(\mathbf{R}) &= \sum_{i=1}^{d} \sqrt{c(\gamma_i)} (\mathbf{S}\boldsymbol{\xi}_i)^H \mathbf{W}_{\boldsymbol{\Gamma}}\boldsymbol{\xi}_i= \tr\left( c(\boldsymbol{\Gamma})^{1/2} \mathbf{S}^H \mathbf{W}_{\boldsymbol{\Gamma}} \right)\\
&= \left\langle \mathbf{S} c(\boldsymbol{\Gamma})^{1/2}, \mathbf{W}_{\boldsymbol{\Gamma}} \right\rangle_F = \left\langle \mathbf{S} c(\boldsymbol{\Gamma})^{1/2}, \mathcal{D}_{\boldsymbol{\Gamma}}^{-1} \left( \mathbf{S} c(\boldsymbol{\Gamma})^{1/2} \right) \right\rangle_F \\
&= \mathbf{a}_{\boldsymbol{\Gamma}}^H \mathbf{D}_{\boldsymbol{\Gamma}}^{-1} \mathbf{a}_{\boldsymbol{\Gamma}}.
\end{aligned}}
\label{eq:operator_directional_equivalence}
\end{equation}
The last equality follows from Corollary~\ref{cor:mRIT_equ_repre}. 
Combining \eqref{eq:operator_directional_equivalence} with Lemma~\ref{lem:mRIT_LB} gives
\begin{equation}
\log\det(\mathbf{I}+\mathbf{R})
=
\max_{\boldsymbol{\Gamma}\succ\mathbf{0}}
\mathbf{a}_{\boldsymbol{\Gamma}}^{H}
\mathbf{D}_{\boldsymbol{\Gamma}}^{-1}
\mathbf{a}_{\boldsymbol{\Gamma}},
\label{eq:single_matrix_rit_identity}
\end{equation}
with the maximizing $\boldsymbol{\Gamma}$ characterized exactly as stated in Theorem~\ref{thm:matrix_rit}.
Applying \eqref{eq:single_matrix_rit_identity} independently to each term of
\eqref{prob:generic_logdet} then summing over $m$ with weights $\omega_m$ yields~\eqref{eq:matrix_rit_problem}.
This completes the constructive proof of Theorem~\ref{thm:matrix_rit}.

\section{Proof of Corollary~\ref{cor:mRIT_equ_repre}}\label{pro:coro_oper_repr}
We first verify that $\mathcal{D}_{\boldsymbol{\Gamma}}^{-1}$ exists.
For arbitrary $\mathbf{X},\mathbf{Y}\in\mathbb{C}^{n\times d}$, the Hermitian properties of \(\mathbf{F}\), \(\boldsymbol{\Gamma}^2\), \(\mathbf{S}\mathbf{S}^{H}\), and \(b(\boldsymbol{\Gamma})\), together with cyclic invariance of the trace \cite[Eq.~(16)]{petersen2012matrix}, give
\begin{equation}
\left\langle
\mathbf{X},\mathcal{D}_{\boldsymbol{\Gamma}}(\mathbf{Y})
\right\rangle_F
=
\left\langle
\mathcal{D}_{\boldsymbol{\Gamma}}(\mathbf{X}),\mathbf{Y}
\right\rangle_F,
\end{equation}
so that \(\mathcal{D}_{\boldsymbol{\Gamma}}\) is self-adjoint. Moreover, for any $\mathbf{A}\neq\mathbf{0}$,
\begin{align}
\left\langle\mathbf{A},\mathcal{D}_{\boldsymbol{\Gamma}}(\mathbf{A})\right\rangle_F
&=
\operatorname{tr}(\mathbf{A}^{H}\mathbf{F}\mathbf{A}\boldsymbol{\Gamma}^{2})
+\operatorname{tr}(\mathbf{A}^{H}\mathbf{S}\mathbf{S}^{H}\mathbf{A}b(\boldsymbol{\Gamma}))\nonumber\\
&=
\|\mathbf{F}^{1/2}\mathbf{A}\boldsymbol{\Gamma}\|_F^2
+\|\mathbf{S}^{H}\mathbf{A}b(\boldsymbol{\Gamma})^{1/2}\|_F^2
>0,
\label{eq:D_positive_energy}
\end{align}
where the strict inequality follows from $\mathbf{F}\succ\mathbf{0}$ and $\boldsymbol{\Gamma}\succ\mathbf{0}$.
Hence $\mathcal{D}_{\boldsymbol{\Gamma}}$ is strictly positive definite and therefore invertible.

Next, by the definition of $\mathbf{D}_{m,\boldsymbol{\Gamma}_m}$ in~\eqref{eq:D_gamma_kron} and $\mathcal{D}_{m,\boldsymbol{\Gamma}_m}$ in~\eqref{eq:Dm_operator}, and using the identity
$\operatorname{vec}(\mathbf{A}\mathbf{Y}\mathbf{B})
=(\mathbf{B}^{T}\otimes\mathbf{A})\operatorname{vec}(\mathbf{Y})$~\cite[Eq.~(520)]{petersen2012matrix}, we obtain
\begin{equation}\label{eq:vec_operator_relation}
\operatorname{vec}\!\left(
\mathcal{D}_{m,\boldsymbol{\Gamma}_m}(\mathbf{Y})
\right)
=
\mathbf{D}_{m,\boldsymbol{\Gamma}_m}\operatorname{vec}(\mathbf{Y}).
\end{equation}
Since both representations are invertible, we have
\begin{equation}\label{eq:vec_inverse_operator_relation}
\operatorname{vec}\!\left(
\mathcal{D}_{m,\boldsymbol{\Gamma}_m}^{-1}(\mathbf{Y})
\right)
=
\mathbf{D}_{m,\boldsymbol{\Gamma}_m}^{-1}\operatorname{vec}(\mathbf{Y}).
\end{equation}
Setting
\(\mathbf{Y}=\mathbf{S}_m c(\boldsymbol{\Gamma}_m)^{1/2}\) and using the identity
$\langle\mathbf{A},\mathbf{B}\rangle_F
=\operatorname{vec}(\mathbf{A})^{H}\operatorname{vec}(\mathbf{B})$~\cite[Eq.~(521)]{petersen2012matrix}, we have
\begin{align}
&\left\langle
\mathbf{S}_m c(\boldsymbol{\Gamma}_m)^{1/2},
\mathcal{D}_{m,\boldsymbol{\Gamma}_m}^{-1}
\!\left(\mathbf{S}_m c(\boldsymbol{\Gamma}_m)^{1/2}\right)
\right\rangle_F \nonumber\\
&\qquad =
\mathbf{a}_{m,\boldsymbol{\Gamma}_m}^{H}
\mathbf{D}_{m,\boldsymbol{\Gamma}_m}^{-1}
\mathbf{a}_{m,\boldsymbol{\Gamma}_m}.
\label{eq:operator_vector_equivalence_single}
\end{align}
Summing \eqref{eq:operator_vector_equivalence_single} over $m$ with weights $\omega_m$ gives \eqref{eq:operator_kron_equivalence} and proves the corollary.

\section{Proof of Lemma~\ref{lem:mRIT_LB}} \label{proof_of_mRIT_LB_lem}
Applying the scalar RIT inequality \eqref{eq:scalar_rit_lower_bound} to every eigenvalue of $\mathbf R$ according to the matrix function property in Section~\ref{subsec:Hermitian functional calculus} gives
\begin{equation}\label{eq:operator_gap}
\phi_{\gamma_i}(\mathbf{R})
\preceq
\log(\mathbf{I}+\mathbf{R}),
\qquad i=1,\ldots,d.
\end{equation}
Taking the Rayleigh quotients of~\eqref{eq:operator_gap} with respect to the orthonormal eigenvectors $\{\boldsymbol{\xi}_i\}_{i=1}^d$ of $\boldsymbol{\Gamma}$ and summing over $i = 1, \dots, d$, we obtain
\begin{align}
\Phi_{\boldsymbol{\Gamma}}(\mathbf{R})&= \sum_{i=1}^{d}
\boldsymbol{\xi}_i^{H}\phi_{\gamma_i}(\mathbf{R})\boldsymbol{\xi}_i
\leq
\sum_{i=1}^{d}
\boldsymbol{\xi}_i^{H}\log(\mathbf{I}+\mathbf{R})\boldsymbol{\xi}_i\nonumber\\
&=
\operatorname{tr}\left(\boldsymbol{\Xi}^{H}\log(\mathbf{I}+\mathbf{R})\boldsymbol{\Xi}\right)
=
\log\det(\mathbf{I}+\mathbf{R}).
\label{eq:matrix_rit_lower_bound_proof}
\end{align}
This proves~\eqref{eq:matrix_rit_lower_bound}. 

If $\mathbf{R}\succ\mathbf{0}$ with the eigenvalue set $\{\lambda_i\}_{i=1}^d,$ choose \(\boldsymbol{\Gamma}=\mathbf{R}\). 
Then each $\gamma_i=\lambda_i$ along the corresponding eigenvector $\boldsymbol{\xi}_i$, and satisfies the scalar RIT tightness condition. This gives $\phi_{\lambda_i}(\lambda_i)=\log(1+\lambda_i)$.
Moreover, when the equality in \eqref{eq:matrix_rit_lower_bound} holds, there is
\begin{equation}
\boldsymbol{\xi}_i^{H}
\left[
\log(\mathbf{I}+\mathbf{R})-\phi_{\gamma_i}(\mathbf{R})
\right]
\boldsymbol{\xi}_i=0,
\qquad i=1,\ldots,d.
\end{equation}
The matrix in brackets is positive semidefinite, and the scalar gap
\(\log(1+t)-\phi_{\gamma_i}(t)\) vanishes for $t>0$ only at
$t=\gamma_i$, that is,
\begin{equation}
\mathbf{R}\boldsymbol{\xi}_i=\gamma_i\boldsymbol{\xi}_i,
\qquad i=1,\ldots,d,
\end{equation}
which implies
\(\mathbf{R}
=\boldsymbol{\Xi}\operatorname{diag}(\gamma_1,\ldots,\gamma_d)\boldsymbol{\Xi}^{H}
=\boldsymbol{\Gamma}\). 
Thus $\boldsymbol{\Gamma}^{\star}=\mathbf{R}$ is unique when
$\mathbf{R}\succ\mathbf{0}$. 

If $\mathbf{R}$ is singular, the choice in
\eqref{eq:gamma_opt_singular} attains equality because
$\phi_{\eta}(0)=0=\log(1+0)$ for every $\eta>0$ and the remaining parts are identical with $\mathbf{R}\succ\mathbf{0}$.

\section{Proof of Theorem~\ref{thm:matrix_sefp}}\label{proof_of_Thm_mSEFP}
Applying Lemma~\ref{lem:operator_qt} term by term to \eqref{eq:operator_kron_equivalence} with $\mathbf{A}=\sqrt{\omega_m}\,\mathbf{S}_m c(\boldsymbol{\Gamma}_m)^{1/2}$ and $\mathcal{D}=\mathcal{D}_{m,\boldsymbol{\Gamma}_m}$,
expanding $\langle\mathbf{Y}_m,
\mathcal{D}_{m,\boldsymbol{\Gamma}_m}(\mathbf{Y}_m)\rangle_F$ and following the definition in~\eqref{eq:Dm_operator}, we have the objective in \eqref{def_of_FmSEFP}.
Therefore,
\begin{equation}\label{eq:appendix_max_y}
\max_{\{\mathbf{Y}_m\}_{m=1}^{M}}
F_{RQ}(\mathbf{x},\underline{\boldsymbol{\Gamma}},\underline{\mathbf{Y}})=
F_R(\mathbf{x},\underline{\boldsymbol{\Gamma}}).
\end{equation}
Combining \eqref{eq:appendix_max_y} with Theorem~\ref{thm:matrix_rit} proves the equivalence in Theorem~\ref{thm:matrix_sefp}. 
The update \eqref{eq:Y_operator_update} follows directly from \eqref{eq:operator_qt_opt}.

Substituting the definition of $\mathcal{D}_{m,\boldsymbol{\Gamma}_m}$ in \eqref{eq:Dm_operator} into~\eqref{eq:Y_operator_update}, $\mathbf{Y}_m^{\star}$ is then rewritten into a more tractable form as the unique solution of a generalized Sylvester equation
\begin{equation}\label{eq:Y_sylvester}
\mathbf{F}_m\mathbf{Y}_m^{\star}\boldsymbol{\Gamma}_m^2
+\mathbf{S}_m\mathbf{S}_m^{H}\mathbf{Y}_m^{\star}b(\boldsymbol{\Gamma}_m) =
\sqrt{\omega_m}\,\mathbf{S}_m c(\boldsymbol{\Gamma}_m)^{1/2}.
\end{equation}
Finally,~\eqref{eq:Y_closeform} can be directly obtained by substituting~\eqref{eq:gamma_opt_pd} into~\eqref{eq:Y_sylvester} for $\mathbf{R}_m$ is non-singular, and by substituting~\eqref{eq:gamma_opt_singular} into~\eqref{eq:Y_sylvester} for $\mathbf{R}_m$ is singular.
\section{Proof of Theorem~\ref{SEFP-XMMSE}}\label{proof_of_thm_SEFP-XMMSE}
Substituting~\eqref{eq:structured_receiver_path} into the $m$-th term of~\eqref{eq:xmmse_functional_objective} and applying Lemma~\ref{lem:kernel_moments} gives
\begin{equation}\label{eq:xmmse_sefp_three_moments}
\autoeq{
\begin{aligned}
&\omega_m\!\int_0^1\!
2\Re\operatorname{tr}(\mathbf S_m^H\mathbf U_{m,\tau}) d\tau
=2\sqrt{\omega_m}\Re\operatorname{tr}\big(
c(\boldsymbol{\Gamma}_m)^{1/2}\mathbf S_m^H\mathbf Y_m\big),\\
&\omega_m\!\int_0^1\!
\operatorname{tr}(\mathbf U_{m,\tau}^H\mathbf F_m\mathbf U_{m,\tau}) d\tau
=\operatorname{tr}(\mathbf Y_m^H\mathbf F_m\mathbf Y_m\boldsymbol{\Gamma}_m^2),\\
&\omega_m\!\int_0^1\!\tau
\operatorname{tr}(\mathbf U_{m,\tau}^H\mathbf S_m\mathbf S_m^H\mathbf U_{m,\tau}) d\tau
=\operatorname{tr}\big(
\mathbf Y_m^H\mathbf S_m\mathbf S_m^H\mathbf Y_m b(\boldsymbol{\Gamma}_m)\big).
\end{aligned}}
\end{equation}
Summing~\eqref{eq:xmmse_sefp_three_moments} over $m$ proves~\eqref{eq:xmmse_sefp_objective_identity}.

Next, substituting~\eqref{eq:xmmse_sefp_optimal_mapping} into~\eqref{eq:structured_receiver_path} yields, for every $m$ and $\tau\in[0,1]$, that
\begin{align}
\mathbf U_{m,\tau}(\boldsymbol{\Gamma}_m^\star,\mathbf Y_m^\star)
&=\mathbf F_m^{-1}\mathbf S_m\mathbf R_m^{-1}
(\mathbf I_{d_m}+\mathbf R_m)^{-1/2}
\nonumber\\
&\quad\times
\mathbf R_m(\mathbf I_{d_m}+\tau\mathbf R_m)^{-1}
(\mathbf I_{d_m}+\mathbf R_m)^{1/2}
\nonumber\\
&=\mathbf F_m^{-1}\mathbf S_m
(\mathbf I_{d_m}+\tau\mathbf R_m)^{-1},
\label{eq:recover_xmmse_receiver}
\end{align}
which is exactly the unique pointwise XMMSE-induced auxiliary variable in~\eqref{eq:xmmse_optimal_receiver}.  Hence restricting $\underline{\mathcal U}$ to the structured family~\eqref{eq:structured_receiver_path} causes no loss at the optimum.  Combining this observation with~\eqref{eq:xmmse_sefp_objective_identity} and the pointwise SEFP equivalence in~\eqref{eq:mm_value_representation} proves~\eqref{eq:xmmse_sefp_lossless}.

\newpage

\vfill

\end{document}